\documentclass[12pt,onecolumn,draftclsnofoot,journal]{IEEEtran}
\usepackage{graphicx}
\usepackage{float}
\usepackage{epstopdf}
\usepackage[cmex10]{amsmath}
\usepackage{array}
\usepackage{cite}
\usepackage{amssymb}
\usepackage{amsfonts}
\usepackage{amsmath}
\usepackage{stackrel}
\usepackage{booktabs,multirow}
\usepackage{arydshln}
\usepackage{amsthm}
\usepackage{listings}
\usepackage{algorithm}
\usepackage{algorithmicx}
\usepackage{algpseudocode}
\usepackage{threeparttable}
\newcommand\numberthis{\addtocounter{equation}{1}\tag{\theequation}}

\newtheorem{lem}{Lemma}
\newtheorem{rem}{Remark}
\newtheorem{theo}{Theorem}

\newtheorem{cor}{Corollary}

\newtheorem{ex}{Example}
\newfloat{routine}{htbp}{loa}
\floatname{routine}{Routine}

\makeatletter
\newcommand{\algmargin}{\the\ALG@thistlm}
\makeatother
\newlength{\forwidth}
\algdef{SE}[parFOR]{parFor}{EndparFor}[1]
  {\parbox[t]{\dimexpr\linewidth-\algmargin}{
     \hangindent\forwidth\strut\algorithmicfor\ #1\ \algorithmicdo\strut}}{\algorithmicend\ \algorithmicfor}
\algnewcommand{\parState}[1]{\State
  \parbox[t]{\dimexpr\linewidth-\algmargin}{\strut #1\strut}}

\newlength{\ifwidth}
\algdef{SE}[parIF]{parIf}{EndparIf}[1]
  {\parbox[t]{\dimexpr\linewidth-\algmargin}{
     \hangindent\ifwidth\strut\algorithmicif\ #1\ \algorithmicdo\strut}}{\algorithmicend\ \algorithmicif}
\begin{document}

\title{On the Tanner Graph Cycle Distribution of Random LDPC, Random Protograph-Based LDPC, and Random Quasi-Cyclic LDPC Code Ensembles}
\author{\IEEEauthorblockN{Ali Dehghan  and Amir H. Banihashemi,\IEEEmembership{ Senior Member, IEEE}}
	\IEEEauthorblockA{\\ Department of Systems and Computer Engineering, Carleton University, Ottawa, Ontario, Canada}
}

\maketitle

%\IEEEpeerreviewmaketitle

\begin{abstract}
In this paper, we study the cycle distribution of random low-density parity-check (LDPC) codes, randomly constructed protograph-based LDPC codes, and random quasi-cyclic (QC) LDPC codes.
%Random bipartite graphs, random lifts of bipartite protographs, and random cyclic lifts of bipartite protographs are used to represent random low-density parity-check (LDPC) codes, randomly constructed protograph-based LDPC codes, and random quasi-cyclic (QC) LDPC codes, respectively. In this paper, we study the distribution of cycles of different length in all three categories of graphs.
We prove that for a random bipartite graph, with a given (irregular) degree distribution, the distributions of cycles of different length tend to independent Poisson distributions, as the size of the graph tends to infinity.
We derive asymptotic upper and lower bounds on the expected values of the Poisson distributions that are independent of the size of the graph, and only depend on the degree distribution and the cycle length.
%We also derive the expected values of the Poisson distributions, and show that they are independent of the size of the graph, and only depend on the degree distribution and the cycle length.
For a random lift of a bi-regular protograph, we prove that the asymptotic cycle distributions are essentially the same as those of random bipartite graphs as long as the degree distributions are identical.
%It is well-known that for a random lift of a protograph, the distributions of cycles of different length $c$ tend to independent Poisson distributions with expected value equal to the number of tailless backtrackless closed (TBC) walks of length $c$ in the protograph, as the size of the graph (lifting degree) tends to infinity. Here, we find the number of TBC walks in a bi-regular protograph, and demonstrate that random bi-regular LDPC codes have essentially the same cycle distribution as random protograph-based LDPC codes, as long as the degree distributions are identical.
For random QC-LDPC codes, however, we show that the cycle distribution can be quite different from the other two categories.
In particular, depending on the protograph and the value of $c$, the expected number of cycles of length $c$, in this case, can be either
$\Theta(N)$ or $\Theta(1)$, where $N$ is the lifting degree (code length).
%While for the former categories, the expected number of cycles of different length is $\Theta(1)$ with respect to the size of the graph, for the case of QC-LDPC codes, depending on the protograph and the value of $c$, it can be either $\Theta(N)$ or $\Theta(1)$, where $N$ is the lifting degree (code length).
%For QC-LDPC codes, we also derive an upper bound on the variance of the number of cycles of different length. This bound increases linearly with $N$.
We also provide numerical results that match our theoretical derivations. Our results provide a theoretical foundation for emperical results that were reported in the literature but were not well-justified.
They can also be used for the analysis and design of LDPC codes and associated algorithms that are based on cycles.

\begin{flushleft}

\end{flushleft}

\end{abstract}

\section{introduction}

The performance of low-density parity-check (LDPC) codes under iterative message-passing algorithms is highly dependent on the structure of the code's Tanner graph, in general, and the distribution of short cycles, in particular, see, e.g.,~\cite{mao2001heuristic},~\cite{hu2005regular},~\cite{halford2006algorithm},~\cite{xiao2009error}.
The cycles play a particularly important role in the error floor performance of LDPC codes, where they form the main substructure of the trapping sets~\cite{asvadi2011lowering},~\cite{MR3252383},~\cite{MR2991821},~\cite{HB-CL},~\cite{HB-IT1}.

Counting and enumerating (finding) cycles of a given length in a general graph is known to be NP-hard \cite{flum2004parameterized}. (For a rather comprehensive literature
review on algorithms to count and enumerate cycles in different types of graphs, including bipartite graphs, and their complexity, the reader is referred to~\cite{karimi2013message}.)
It is thus of interest to have simple approximations for the number of cycles of a given length in a given graph. Related to this, it is also interesting to
obtain the distribution of cycles of a given length in an ensemble of Tanner graphs (LDPC codes). The knowledge of such a distribution, including the
expected value and variance, can help in the analysis and in guiding the design of LDPC codes. The expected value can also be used
as an approximation for the number of cycles of a given length in a given graph in the ensemble, with the variance providing a measure of accuracy of
the approximation.

In \cite{bollobas1980probabilistic}, Bollob{\'a}s showed that, for a given random graph with an arbitrary degree distribution and a fixed $c$, as the size of the graph
tends to infinity, the multiplicities of cycles of lengths $3,4,5,\ldots,c$, tend to independent Poisson random variables. He also derived the expected values of the random variables. Later, in~\cite{mckay2004short}, the authors considered random bipartite graphs, in which all the nodes have the same
degree $d$, and  $c$ can grow as a function of the number nodes in the graph, and proved that as the size of the graph tends to infinity, the distributions of cycles of different length $c$ tend to independent
Poisson distributions with expected values ${\mu} = (d-1)^{c} /c$.

In this work, we consider the case of random bipartite graphs with arbitrary
degree distributions $\{d_i\}$ and $\{d'_i\}$ on the two parts of the graph, respectively, and prove that the multiplicities of cycles of different length $c$, as the size of the graph tends to infinity, tend to independent Poisson random variables with the following expected values:
\begin{equation}
\mu \approx \dfrac{\displaystyle\Big((\frac{2}{|E|}\displaystyle\sum_{i=1}^{n}{{d_i}\choose {2}}) (\frac{2}{|E|}\displaystyle\sum_{i=1}^{m}{{d_i'}\choose {2}})\Big)^{c/2}}{c}\:,
\label{eqwg}
\end{equation}
where $n$ and $m$ are the number of nodes in the two parts of the graph, and $|E|$ is the number of edges of the graph. The notation ``$\approx$'' in (\ref{eqwg}) is used to mean ``approximately equal,'' and the approximation is within some fixed multiplicative factor of the exact value. Unlike the bipartite graphs studied in \cite{mckay2004short},
the graphs studied in this work are those representing (irregular and bi-regular) LDPC codes.

For the special case of bi-regular LDPC codes, Equation (\ref{eqwg})
reduces to
\begin{equation}
\mu \sim \dfrac{\Big((d_u-1)(d_w-1)\Big)^{c/2}}{c}\:,
\label{eqhg}
\end{equation}
in which $d_u$ and $d_w$ denote the degrees of nodes in the two parts of the graph. The notation ``$\sim$'' in (\ref{eqhg}) is used to mean ``asymptotically equal.'' Equation (\ref{eqhg}) implies that, at sufficiently large block lengths, the average number of cycles, as well as the variances, do not depend on the block length of the code. This matches the observation made in~\cite{karimi2013message}
through numerical results.

The construction of LDPC codes by lifting a small bipartite graph, called {\em base graph} or {\em protograph}, was first appeared in~\cite{thorpe2003low}. Sine then, there has been a flurry of research activity on the analysis and design of protograph-based LDPC codes, see, e.g.,~\cite{MR2810289},~\cite{MR3225941}, and the references therein. A particularly popular category of protograph-based LDPC codes are those constructed by {\em cyclic liftings}~\cite{fossorier2004quasicyclic},~\cite{MR2236257},~\cite{kim2007cycle},~\cite{karimi2012counting},~\cite{MR3124643},
~\cite{MR3071345},~\cite{MR3252383}. Such codes are quasi cyclic (QC), and are of most interest in practice, as they lend themselves to
%a more compact representation, which in turn, translates to
simpler implementation of encoding and decoding algorithms. For that reason, they have also been adopted in a number of standards \cite{s1,s2}.

%The set of protograph-based LDPC codes is a class of LDPC codes lifted from protographs.
It was shown by Fortin and Rudinsky in \cite{fortin2012asymptotic} that for a random lift of a protograph,
the distributions of cycles of different length tend to independent Poisson distributions as
the size of the graph tends to infinity. They also showed that the expected value of the number of cycles of length $c$ is equal to $T(G,c)$, where $T(G,c)$ is the number of tailless backtrackless closed walks of length $c$ in the protograph $G$. In this work, we calculate $T(G,c)$ for bi-regular protographs, in general, and fully-connected bipartite protographs, in particular. Using these results, we show that the cycle distributions of random bi-regular graphs and those of random lifts of a bi-regular protograph with a similar degree distribution are essentially identical in the asymptotic regime, where the graph size tends to infinity.

In \cite{karimi2012counting}, an efficient algorithm for counting short cycles in the Tanner graph of a QC-LDPC code is proposed. Using numerical results, it was shown in \cite{karimi2012counting}, that randomly constructed QC-LDPC codes
have a much better girth distribution compared to their counterparts that lack the QC structure.
%
%A   cycle  in the Tanner graph of QC-LDPC codes is called an  inevitable cycle, if that cycle always exits regardless of lifting degree $N$ (code's block length) and the permutation shifts. This type of cycles  received a great interest during recent years, for instance see \cite{kim2007cycle}, \cite{MR3124643}, \cite{MR2236257}. Inevitable cycles play an important role in the QC-LDPC codes. Through our results for QC-LDPC codes on the average number of cycles of length $c$, we show that if a random QC bipartite graph contains at least one inevitable cycle of length $c$ then, when $N$ is a large number, among the set of cycles of length $c$, the set of inevitable cycles of length $c$ is dominant.
%
In this work, by viewing the Tanner graphs of QC-LDPC codes as cyclic lifts of protographs, we study their cycle distribution. We demonstrate that the cycle distributions for random cyclic lifts of a bipartite protograph can be quite different from those of random bipartite graphs and random lifts of bipartite protographs of similar degree distributions.
In particular, we show that depending on the protograph and the cycle length $c$, the expected value of the number of cycles of length $c$ in random cyclic lifts can increase linearly with the size of the graph. This is while for random bipartite graphs and random lifts of bipartite protographs, the expected number of cycles of length $c$ remains constant with increase in the graph size, regardless of the value of $c$ or the choice of protograph or degree distribution.
These results explain the differences observed in \cite{karimi2012counting} regarding the
cycle distributions of QC-LDPC codes versus LDPC codes that lack the QC structure.

In addition to providing theoretical justification for empirical results in the literature, the results presented here can be used for the analysis and design of LDPC codes and associated algorithms that are based on cycles.
As an example, it was shown very recently~\cite{IT2cycle} that among trapping set structures with cycles, only those that contain a single (chordless) cycle have non-zero multiplicity asymptotically.
This asymptotic multiplicity has been estimated in \cite{IT2cycle} using the results of this work. (More details are provided in Subsection~\ref{subsec2.2}.)
As another example, the $dpl$ characterization and search algorithm of~\cite{HB-IT1},~\cite{Y-Arxiv} is known to be the most efficient in exhaustively finding the
elementary trapping sets of LDPC codes. The starting point of $dpl$ search is chordless cycles in a graph. These cycles are then recursively expanded using three simple expansion techniques.
Our theoretical results on the average number of cycles can be used to establish theoretical bounds on the average complexity of $dpl$ search.
In the absence of such theoretical results, complexity discussions in~\cite{HB-IT1}, related to the number of cycles, relied on empirical results provided in~\cite{karimi2013message}.

The organization of the rest of the paper is as follows: In Section~\ref{sec1}, we present some definitions and notations. This is followed in Section~\ref{sec2} by our results
on the cycle distribution of random LDPC codes. In this section, we also present an application of our results to estimate the asymptotic multiplicity of trapping sets.
In Section \ref{secN1}, we discuss the cycle distribution of random lifts of a protograph, and calculate $T(G,c)$ for
bi-regular protographs. The results on the expected value and the variance of the number of cycles for QC-LDPC codes are presented in Section~\ref{sec3}.
Section~\ref{sec4} is devoted to numerical results. The paper is concluded with some remarks in Section~\ref{sec5}.

\section{Definitions and notations}
\label{sec1}

An undirected graph $G = (V,E)$ is defined as a set of vertices or nodes $V$ and a set of edges $E$, where $E$ is a subset of
the pairs $\{\{u,v\}: u,v\in V , u\neq v\}$. In this work, we consider graphs with
no loop or parallel edges.  A  graph is called {\em complete}
if every node is connected to all the other nodes. We use the notation $K_a$ for a complete graph with $a$ nodes.
A {\it walk} of length $k$ in the graph $G$ is a sequence of nodes
$v_1, v_2, \ldots , v_{k+1}$ in $V$ such that $\{v_i, v_{i+1}\} \in E$, for all $i \in \{1, \ldots , k\}$. Equivalently, a walk of length $k$ can be described
by the corresponding sequence of $k$ edges. A walk is a {\it path} if all the nodes $v_1, v_2, \ldots , v_k$ are distinct. A walk is called a
{\it closed walk}  if the two end nodes are identical, i.e.,
if $v_1 = v_{k+1}$. Under the same condition, a path is called a {\it cycle}.  We call a cycle {\em chordless} if no two nodes of the cycle are connected by an edge that does not itself belong to the cycle. Otherwise, such an edge is called a {\em chord} of the cycle. We denote cycles of length $k$, also referred to as $k$-cycles,
by $C_k$. We use $N_k$ for $|C_k|$. The length of the shortest cycle in a graph is called {\em girth}.

Consider a walk ${\cal W}$ of length $k$ represented by the sequence of edges $e_{i_1}, e_{i_2}, \ldots , e_{i_k}$. The
walk ${\cal W}$ is {\it backtrackless}, if $e_{i_s}  \neq e_{i_{s+1}}$, for any $s \in\{1, \ldots , k-1\}$.
Also, the walk ${\cal W}$ is {\it tailless}, if $e_{i_1}  \neq e_{i_{k}}$. In
this paper, we  use the term {\it TBC walk} to refer to a tailless backtrackless closed walk.
%We call a TBC walk $w$ a {\em multiple} of another TBC walk $w'$, if $w'$ is a subwalk of $w$, and $w$ consists of multiple passes over $w'$.

The {\it adjacency matrix} of a graph $G$ is the matrix $A = [a_{ij}]$, where $a_{ij}$ is the number of edges connecting the node $i$ to the node
$j$ for all $i, j\in V$. Matrix $A$ is symmetric and since we have assumed that $G$ has no parallel edges or loops, $a_{ij}\in\{0, 1\}$
for all $i, j\in V$, and $a_{ii} = 0$ for all $i \in V $.
%The number of closed walks of length $k$ in $G$ is $tr(A^k)/k$, where $tr(\cdot)$ is the trace of a matrix.
One important property of the adjacency matrix that we will use for our results is that the number of walks
between any two nodes of the graph can be determined using the powers of this matrix. More precisely, the entry in
the $i^{\text{th}}$ row and the $j^{\text{th}}$ column of $A^k$, $[A^k]_{ij}$ , is the number of walks of length $k$ between nodes $i$ and $j$. In particular, $[A^k]_{ii}$
is the number of closed walks of length $k$ containing node $i$.

A graph $G=(V,E)$ is called {\it bipartite}, if the node set $V$ can be
partitioned into two disjoint subsets $U$ and $W$, i.e., $V = U \cup W \text{ and } U \cap W =\emptyset $, such that every edge in $E$ connects a node
from $U$ to a node from $W$. Tanner graphs of LDPC codes are bipartite graphs, in
which $U$ and $W$ are referred to as {\it variable nodes} and {\it check
nodes}, respectively. Parameters $n$ and $m$ in this case are used to denote $|U|$ and $|W|$, respectively. Parameter $n$ is the code's block length
and the code rate $R$ satisfies $R \geq 1- (m/n)$.

The number of edges connected to a node $v$ is called the {\em degree} of the node $v$, and is denoted by $d_v$ (or $deg(v)$). We
call a bipartite graph $G = (U\cup W,E)$ {\it bi-regular}, if all the nodes on the same side of the given bipartition have the same degree,
i.e., if all the nodes in $U$ have the same degree $d_u$ and all the nodes in $W$ have the same degree $d_w$.
Note that, for a bi-regular graph, $|U|d_u=|W|d_w=|E|$.
%, where $|E|$ is the number of edges in the bipartite graph.
A bipartite graph that is not bi-regular is called {\it irregular}. We call a bipartite graph
{\em fully-connected} or {\em complete}, if it is bi-regular and if $d_u = |W|$ and $d_w = |U|$.
%A (non-bipartite) graph is called {\em complete} if every node is connected to all the other nodes. We use the notation $K_a$ for a complete graph with $a$ nodes.

%Consider the set of integer numbers $\{0,1, 2, \ldots , N-1\}$, denoted by $\mathbb{Z}_N$.
%Also, consider the subgroup $C_N$ of symmetric group $S_N$ over $\mathbb{Z}_N$, where
%$C_N$ contains all circulant permutations $\pi_p$. The index $p$ of the permutation $\pi_p$ corresponds to $p$ cyclic shifts to the left.
Let  $G(V=U \cup W,E)$ be a bipartite graph with $|U|=n'$ and $|W|=m'$, and consider an assignment of a permutation $ \pi^e \in S_N$ to each edge $e$ in $E$,
where $S_N$ is the symmetric group over $\mathbb{Z}_N = \{0,1, 2, \ldots , N-1\}$.
%circulant permutation $ \pi^e \in C_N$ to each edge $e$ in $E$. Now
Consider the following construction of the graph $\tilde{G}(\tilde{V},\tilde{E})$ from $G(V,E)$: We make $N$ copies of $G$ such that for each
node $v \in V$, we have a set of nodes  $\tilde{v}=\{v^0, \ldots, v^{N-1}\}$ in $\tilde{V}$. Similarly, for each edge $e = \{u,w\}\in E$,
we have a set of edges $\tilde{e}=\{e^0, \ldots, e^{N-1}\}$ in $\tilde{E}$ such that $\{u^i, w^j\}$ belongs
to $\tilde{E}$ if and only if $ \pi^e(i)=j$.
In this construction, graph $\tilde{G}$ is called an $N$-{\em lifting} of $G$. Graph $G$ is called the {\em base graph} or {\em protograph}, and the parameter $N$ is referred to as
the {\em lifting degree}.
%Throughout the paper we will use the terms ``base graph'' and ``protograph,'' interchangeably.
The lifted graph $\tilde{G}$ can be considered as the Tanner graph of an LDPC code $\tilde{C}$, i.e.,
%In particular, the conventional QC-LDPC codes are cyclic liftings of the complete (fully-connected) protograph.
the parity-check matrix $\tilde{H}$ of $\tilde{C}$ is defined to be the incidence matrix of $\tilde{G}$. The code $\tilde{C}$, in this case, is called
the {\it lifted code}, and the incidence matrix $H$ of $G$ is called the {\it base matrix}.
The $m'N \times n'N$ parity-check matrix $\tilde{H}$ of $\tilde{C}$ consists
of $m' \times n'$ submatrices $[ \tilde{H}]_{ij}$, $0 \leq i \leq m'-1$, $0 \leq j \leq n'-1$, where each submatrix is
a permutation matrix of size $N \times N$, if the entry $[H]_{ij} \neq 0$; otherwise, $[ \tilde{H}]_{ij}$ is the all-zero matrix. The LDPC codes constructed by the lifting process, just explained, are
referred to as {\it protograph-based LDPC codes}.
In the lifting process, if the permutations are selected randomly from $S_N$, the constructed codes are called {\em random lifts}.

Consider the subgroup $C_N$ of symmetric group $S_N$ over $\mathbb{Z}_N$, where $C_N$ contains all circulant permutations $\pi_p$.
The index $p$ of the permutation $\pi_p$ corresponds to $p$ cyclic shifts to the left.
If the permutations in the lifting process are cyclic, i.e., if they are selected from $C_N$, then the resulting graph $\tilde{G}$ is called a {\em cyclic lift} of $G$, and
the associated code is quasi-cyclic (QC). In this case, the non-zero submatrices of $\tilde{H}$ are circulant permutation matrices (CPM). In particular, when the entry $[H]_{ij}  \neq 0$, then
$[ \tilde{H}]_{ij} = I^{p_{ij}}$, $p_{ij} \in \mathbb{Z}_N$, where $I^{p_{ij}}$ is a CPM whose rows are
obtained by cyclically shifting the rows of the identity matrix
to the left by $p_{ij}$. We also take $I^{+\infty}$ to represent the all-zero matrix. We refer to the $m'\times n'$ matrix $P = [p_{ij}]; 0 \leq i \leq m'-1, 0 \leq j \leq n'-1$, as the
{\em permutation shift matrix} or the {\em exponent matrix} corresponding
to the lifted code $\tilde{C}$ or to the lifted graph $\tilde{G}$. Clearly, there is
a one-to-one correspondence between $P$ and $\tilde{H}$.

Consider a QC-LDPC code $\tilde{C}$ corresponding to an exponent matrix $P$.
%
%$$ \tilde{H}=
% \begin{pmatrix}
%  I^{p_{00}}    & I^{p_{01}}    & \cdots & I^{p_{0(n-1)}} \\
% I^{p_{10}}    & I^{p_{11}}    & \cdots & I^{p_{1(n-1)}} \\
% \vdots     & \vdots       & \ddots & \vdots  \\
% I^{p_{(m-1)0}}    & I^{p_{(m-1)1}}     & \cdots & I^{p_{(m-1)(n-1)}}
% \end{pmatrix},
%$$
%
%where $p_{ij}\in \mathbb{Z}_N$ for $0 \leq i \leq m-1$, $0 \leq j \leq n-1$.
%The matrix $\tilde{H}$ corresponds to the following exponent matrix
%
%$$ P=
% \begin{pmatrix}
%   {p_{00}}    &  {p_{01}}    & \cdots &  {p_{0(n-1)}} \\
%   {p_{10}}    &  {p_{11}}    & \cdots &  {p_{1(n-1)}} \\
%  \vdots     & \vdots       & \ddots & \vdots  \\
%   {p_{(m-1)0}}    &  {p_{(m-1)1}}     & \cdots &  {p_{(m-1)(n-1)}}
% \end{pmatrix}.
%$$
%
It is well-known  that a necessary  condition
for the existence of a cycle of length $2k$ in the Tanner graph
of $\tilde{C}$, corresponding to $\tilde{H}$, is
\begin{equation}\label{E5}
\displaystyle\sum_{i=0}^{k-1} (p_{m_i,n_i}-p_{m_i,n_{i+1}})=0 \mod N\:,
\end{equation}
where $n_k = n_0$, $m_i \neq m_{i+1}$, $n_i \neq n_{i+1}$, and none of the permutation shifts in (\ref{E5}) is $+\infty$~\cite{fossorier2004quasicyclic}.
The sequence of permutation shifts in (\ref{E5}) corresponds to a TBC walk in the base graph, i.e., cycles of the lifted graph are the inverse images of
TBC walks with zero permutation shift in the base graph~\cite{MR898434}, \cite{MR3071345}. In fact, an additional requirement for the sequence of permutation shifts in (\ref{E5})
to correspond to a cycle in the lifted graph is that no subsequence of the permutation shifts should correspond to a TBC walk of permutation shift zero in the base graph~\cite{MR898434}, \cite{MR3071345}. For a TBC walk $w$ in $G$,
we refer to the summation in (\ref{E5}) as the permutation shift corresponding to $w$ and denote it by ${\cal P}(w)$. It is clear that depending on the starting index $n_0$ of $w$, or the direction of travel along $w$,
the sign of ${\cal P}(w)$ may change. As we are only concerned about the value of ${\cal P}(w)$ being zero or non-zero, in the context of this work, the two values $\pm {\cal P}(w)$ are considered equivalent.

It is well-known that there are cycles in cyclic lifts of a base graph that are independent of the lifting degree $N$ or the choice of the exponent matrix $P$~\cite{fossorier2004quasicyclic},~\cite{MR2236257}.
Such cycles, referred to as {\it inevitable cycles}, occur if there exists a TBC walk $w$ in the base graph, in which, each edge is traversed in both directions equal number of times. In this case, ${\cal P}(w) = 0 \mod N$,
regardless of the value of $N$, or the choice of $P$. Such TBC walks are referred to as {\em zero-permutation (ZP) TBC walks} in this paper. We also use the terminology {\em prime ZP TBC walk} for a ZP TBC walk
that does not contain any ZP TBC subwalk. In fact, inevitable cycles in the lifted graph are the inverse images of prime ZP TBC walks in the base graph.
Clearly, inevitable cycles (prime ZP TBC walks) only depend on the structure of the base graph.
%A   cycle  in the Tanner graph of $\tilde{C}$, corresponding to $\tilde{H}$, is called an {\it inevitable cycle}, if that cycle always exits regardless of
%$N$ and the permutation shifts. It is clear that the inevitable cycles are caused by the structure of the base graph.

\section{Random Irregular and Bi-Regular  Graphs}
\label{sec2}

\subsection{Main Result}

In the following, we prove our  result on the cycle distribution of random irregular bipartite graphs with arbitrary degree distributions. 
%Although many of the general steps of the proof are similar to those taken in \cite{bollobas1980probabilistic} to prove the result on the cycle distribution of a (non-bipartite) random graph, there are major differences in the details.

\begin{theo}\label{T01}
Let $\Delta_u$, $\Delta_w$, $\delta_u$ and $\delta_w$ be fixed natural numbers satisfying $\Delta_u = d_1 \geq d_2  \geq \ldots \geq d_n = \delta_u > 1$, and $\Delta_w = d_1' \geq d_2'  \geq \ldots \geq d_{m}' = \delta_w > 1$, where
$\sum_{i=1}^n d_i =\sum_{i=1}^{m} d_i'= \eta$. 
%Also, let $2\eta - n \rightarrow  \infty $ and $2 \eta - m \rightarrow  \infty $ as $n,m \rightarrow \infty$.
Consider the probability space $\mathcal{G}$ of all bipartite graphs with node set $(U,W)$, where $U=\{u_1, u_2, \ldots , u_n\}, W=\{w_1,w_2, \ldots, w_{m}\}$, and
in which the degree of node $u_i$ is $d_i$ and the degree of node $w_i$ is $d_i'$. Suppose that the graphs in $\mathcal{G}$ are selected uniformly at random.
For $G \in \mathcal{G}$, denote by $N_i(G)$ the number of cycles of length $i$ in $G$. Then, as $n,m \rightarrow \infty$, for any fixed even value of $k \geq 4$, the random variables $N_4, N_6, \ldots, N_k$,
are asymptotically independent Poisson random variables with $N_c$ having the expected value
$$
E(N_c)\approx \dfrac{\displaystyle\Big((\frac{2}{\eta}\displaystyle\sum_{i=1}^{n}{{d_i}\choose {2}}) (\frac{2}{\eta}\displaystyle\sum_{i=1}^{m}{{d_i'}\choose {2}})\Big)^{c/2}}{c}\:,
$$
where the approximation is an asymptotic upper bound within the fixed multiplicative factor of $[S(h_u) \times S(h_w)]^{-c/2}$ from the exact value, with Specht's ratio $S(h)$ defined by $S(h)=\dfrac{(h-1)h^{\frac{1}{h-1}}}{e \log h}$ for $h \neq 1$, and $S(1)=1$ ($e$ is Euler's constant), and $h_u= \frac{\Delta_u(\Delta_u-1)}{\delta_u(\delta_u-1)}$, $h_w= \frac{\Delta_w(\Delta_w-1)}{\delta_w(\delta_w-1)}$.
\end{theo}

%Recall that if $r \in \mathbb{N}$ and $X $ is a random variable, we define the $r$-th factorial moment of $X$ to be $E[(X)_r]=E[X(X-1)\ldots(X-r+1)]$. Also, let $X_1,X_2,\ldots$ and $Z$ be random variables. We say that $X_n$ {\it converges in distribution to} $Z$ as $n \rightarrow\infty$, and write $X_n \xrightarrow{d} Z$, if $\Pr(X_n \leq x)\rightarrow \Pr(Z \leq x)$ for every real $x$ that is continuity point of $\Pr(Z \leq x)$ (see \cite{MR1782847}, pp 8).
To prove the result of Theorem~\ref{T01}, we need a series of intermediate results as discussed below.

We first construct the ensemble ${\cal G}$ of random bipartite graphs, indicated in Theorem~\ref{T01}, in two steps. In the first step, for each node $z$, we consider a bin that contains $deg (z)$ cells.
We then consider random perfect matchings to pair the cells on the $U$ side of the graph to the cells on the $W$ side. The set of all such matchings is denoted by $\Phi$, and we have $|\Phi| = \eta!$,
where $\eta$ is the number of edges in the graph. Corresponding to each matching, there is a so-called {\em configuration}, in which the matched cells on the two sides of the graph 
are connected by an edge. In the rest of the paper, we assume that configurations are selected uniformly at random.
Corresponding to each matching (configuration), we construct a  bipartite graph such that if there is an edge between two cells, then we place an edge between the corresponding nodes (bins) in the  bipartite graph.
The bipartite graphs are thus represented as images of the configurations. We denote the ensemble of bipartite graphs so constructed by 
${\cal G}^*$. We note that ${\cal G}^*$ contains bipartite graphs with parallel edges,\footnote{In the rest of the paper, due to the possibility of parallel edges existing in the bipartite graphs in ${\cal G}^*$, we use the term {\em multigraph} to refer to such graphs. A multigraph is called {\em simple} if it has no parallel edges.} and that a uniform distribution over the configurations induces a non-uniform distribution over the ensemble ${\cal G}^*$. 
The second step in the construction of ${\cal G}$ is to remove all the bipartite graphs with parallel edges from ${\cal G}^*$. It is now straightforward to see that, with the condition that the bipartite graphs constructed from 
random configurations have no parallel edges, the distribution of bipartite graphs (those in ${\cal G}$) is uniform. 
%The set of corresponding configurations (with no parallel edges) is denoted by $\Omega$. We note that a uniform distribution on the configurations in $\Omega$ induces a uniform distribution on ${\cal G}$. 
This is because corresponding to each graph in ${\cal G}$, we have the same number $d_1! \times \cdots \times d_n! \times d'_1! \times \cdots \times d'_m$ of configurations.

%First, we allow multiple edges and work with random multigraphs, afterwards, we convert random multigraphs in our result to random graphs.
%To construct the ensemble ${\cal G}^*$ of random   bipartite multigraphs, 
% Note that ${\cal G}^*$ does not have the uniform distribution over all multigraphs, because different multigraphs arise from  different numbers of configurations.
%Now, we prove the main auxiliary result for random multigraphs.
We now prove a result on the cycle distribution of ${\cal G}^*$ (Theorem~\ref{NewT2}) as an intermediate step to prove Theorem~\ref{T01}. 
To prove the result of Theorem~\ref{NewT2}, we first recall the joint version of Poisson approximation theorem  as follows (see, e.g., \cite{MR1782847}, p. 145).

\begin{theo}\label{NewT1} (Joint version of Poisson approximation theorem)
For each $i \in \{1,2,\ldots,m\}$, consider the sequence of random variables $X_{i,1}, X_{i,2}, \ldots$, each taking values in $\mathbb{N} \cup \{0\}$. 
Suppose there exist $\lambda_1, \lambda_2, \ldots, \lambda_m \in \mathbb{R}_{\geq 0}$, such that for any fixed $r_1, r_2, \ldots, r_m \in \mathbb{N} \cup \{0\}$, we have
\begin{center}
$E[(X_{1,n})_{r_1}(X_{2,n})_{r_2}\ldots (X_{m,n})_{r_m}]\rightarrow \prod_{i=1}^{m}\lambda_{i}^{r_i}$\,\,\, as $n\rightarrow \infty$,
\end{center}
where $(X)_r=X(X-1)\ldots(X-r+1)$, for $r \in \mathbb{N}$, and $(X)_0=1$. 
Then as $n\rightarrow \infty$, the random vector $(X_{1,n},X_{2,n}, \ldots, X_{m,n})$ converges to  
%\xrightarrow{d} 
$(Y_1,Y_2,\ldots,Y_m)$ in distribution, where the random variables $Y_i$ are independent Poisson random variables with $E[Y_i]=\lambda_i$ (i.e., for each $i$, $Y_i$ is $Poisson(\lambda_i)$).
\end{theo}

\begin{theo}\label{NewT2}
Let $\Delta_u$, $\Delta_w$, $\delta_u$ and $\delta_w$ be fixed natural numbers satisfying $\Delta_u = d_1 \geq d_2  \geq \ldots \geq d_n = \delta_u > 1$, and $\Delta_w = d_1' \geq d_2'  \geq \ldots \geq d_{m}' = \delta_w > 1$, where
$\sum_{i=1}^n d_i =\sum_{i=1}^{m} d_i'= \eta$.
Consider the probability space $\mathcal{G}^*$ of all bipartite multigraphs with node set $(U,W)$, where $U=\{u_1, u_2, \ldots , u_n\}, W=\{w_1,w_2, \ldots, w_{m}\}$, and
in which the degree of node $u_i$ is $d_i$ and the degree of node $w_i$ is $d_i'$. (The probability distribution of graphs in $\mathcal{G}^*$  is assumed to be induced by the uniform distribution over configurations.)
For $G \in \mathcal{G}^*$, denote by $N_i(G)$ the number of cycles of length $i$ in $G$. Then, as $n,m \rightarrow \infty$, for any fixed even value of $k \geq 2$, the random variables $ N_2,N_4, \ldots, N_k$,
are asymptotically independent Poisson random variables with $N_c$ having the expected value
$$
E(N_c)\approx  \dfrac{\displaystyle\Big((\frac{2}{\eta}\displaystyle\sum_{i=1}^{n}{{d_i}\choose {2}}) (\frac{2}{\eta}\displaystyle\sum_{i=1}^{m}{{d_i'}\choose {2}})\Big)^{c/2}}{c}\:,
$$
where the approximation is an asymptotic upper bound within a fixed multiplicative factor from the exact value as described in Theorem~\ref{T01}.
\end{theo}

\begin{proof}{
%We use the method of moments (Theorem \ref{NewT1}). 
We start by computing the expectation of $N_c$, also denoted by $\lambda_c$ in the course of the proof. We then apply Theorem~\ref{NewT1} to prove that cycle multiplicities are independent Poisson random variables.%We note that somewhere for simplicity we denote $E(N_c)$ by $\lambda_c$.

\underline{Calculation of $E(N_c)$}. To simplify the calculation of $E(N_c)$, rather than working in the non-uniform probability space of $\mathcal{G}^*$, we perform the calculations in the space of configurations (with uniform distribution).
For a configuration, we define a cycle of length $k$ to be a set of $k$ edges, like $\{e_1, e_2, \ldots , e_k\}$, that connect $k$
distinct bins, like $D_{i_1}, \ldots, D_{i_k}$. The connections are such that for each $j \in \{1,\ldots,k\}$, the edge $e_j$, connects a cell in bin $D_{i_j}$ to a cell in bin $D_{i_{j+1}}$, where $D_{i_{k+1}}=D_{i_{1}}$,
and the two cells in each bin $D_{i_{j}}$, connected to the two edges $e_{j}$ and $e_{j-1}$, are distinct ($e_0 = e_k$).
We now compute the number of $k$-cycles, ${\cal C}_k$, in a configuration. To form a $k$-cycle, one needs to
choose $k/2$ bins from $U$ and $k/2$ bins from $W$. Next, from each bin, one needs to choose two cells (the order of the two cells is important).
Suppose that bin $i$ contains $d_i$ cells. We thus have $(d_i)(d_i-1)$ choices for the two cells of bin $i$. Hence, in order to choose all the cells on both sides of the graph, we have
\begin{equation}\label{NewE1}
\Big( \displaystyle\sum_{\sigma \subset U \atop |\sigma|=k/2}\prod_{u_i \in \sigma} (d_i)(d_i-1) \Big)\Big( \displaystyle\sum_{\sigma \subset W \atop |\sigma|=k/2}\prod_{w_i\in \sigma}(d_i')(d_i'-1) \Big)
\end{equation}
choices. To count the number of $k$-cycles in a configuration, we also need to consider different orderings of the $k/2$ bins on each side of the graph.
This results in
\begin{equation} \label{E130}
 {\cal C}_k =\Big( \displaystyle\sum_{\sigma \subset U \atop |\sigma|=k/2}\prod_{u_i \in \sigma} (d_i)(d_i-1) \Big)\Big( \displaystyle\sum_{\sigma \subset W \atop |\sigma|=k/2}\prod_{w_i\in \sigma}(d_i')(d_i'-1) \Big)\Big( \dfrac{(\frac{k}{2})!(\frac{k}{2})!}{k} \Big)\:,
\end{equation}
where the division by $k$ is for counting each cycle in the above process $k$ times. 
%Note that for the case $k=2$, although $\frac{(\frac{k}{2})!(\frac{k}{2})!}{k}=\frac{1}{2}$, we count each parallel edges twice in (\ref{NewE1}).

We note that given a set of $\ell$ fixed edges, there are $(\eta-\ell)!$ configurations containing those edges.
We then have

\begin{align*}
E(N_c) & = \dfrac{{\cal C}_c \times(\eta-c)!}{\eta!}\\
       &= \Big( \displaystyle\sum_{\sigma\subset U \atop |\sigma|=c/2}\prod_{u_i\in \sigma} (d_i)(d_i-1) \Big)\Big( \displaystyle\sum_{\sigma\subset W \atop |\sigma|=c/2}\prod_{w_i\in \sigma}(d_i')(d_i'-1) \Big)\Big( \dfrac{(\frac{c}{2})!(\frac{c}{2})!}{c} \Big) \times \dfrac{(\eta-c)!}{\eta!}\\
       &\sim \Big( \displaystyle\sum_{\sigma\subset U \atop |\sigma|=c/2}\prod_{u_i\in \sigma}(d_i)(d_i-1) \Big)\Big( \displaystyle\sum_{\sigma\subset W \atop |\sigma|=c/2}\prod_{w_i\in \sigma}(d_i')(d_i'-1) \Big)\Big( \dfrac{(\frac{c}{2})!(\frac{c}{2})!}{c} \Big) \times \dfrac{1}{\eta^c}.  \numberthis \label{E91}
\end{align*}

In the following, we derive asymptotic upper and lower bounds on (\ref{E91}) that differ only in a constant multiplicative factor. For this, we first use Maclaurin's inequality (see \cite{MR3015124}, pp 117-119), as described below.
Let $a_1, a_2, \ldots, a_n$ be positive real numbers, and for $k = 1, 2, \ldots, n$, define the averages $S_k$ as follows:

$$S_k = \frac{\displaystyle \sum_{ 1\leq i_1 < \cdots < i_k \leq n}a_{i_1} a_{i_2} \cdots a_{i_k}}{\displaystyle {n \choose k}},$$
where the summation is over all distinct sets of $k$ indices. Maclaurin's inequality then states:
$$S_1 \geq \sqrt{S_2} \geq \sqrt[3]{S_3} \geq \cdots \geq \sqrt[n]{S_n}\:.$$
Using Maclaurin's inequality, we thus have:
$$S_1 \geq \sqrt[c/2]{S_{\frac{c}{2}}} \geq \sqrt[n]{S_n}\:,$$
or equivalently,
%$$ {\displaystyle {n \choose {c/2}}}S_{\frac{c}{2}} \leq {\displaystyle {n \choose {c/2}}} S_1^{\frac{c}{2}} $$ Hence,
\begin{equation}\label{FNE1}
\sqrt[n]{\prod_{j=1}^{n}a_j} \leq  \sqrt[c/2]{S_{\frac{c}{2}}} \leq  \dfrac{\sum_{j=1}^{n}a_j}{n}\:.
\end{equation}
On the other hand, we have~\cite{specht1960theorie}:
\begin{equation}\label{FNE2}
\frac{1}{S(h)} \times \dfrac{\sum_{j=1}^{n}a_j}{n} \leq \sqrt[n]{\prod_{j=1}^{n}a_j}\:,
\end{equation}
where $h=\frac{M}{m}( \geq 1)$ with $M$ and $m$ equal to the maximum and minimum values of numbers $a_1, a_2, \ldots, a_n$, and Specht's ratio $S(h)$ is defined by 
\begin{equation}\label{FNE3}
S(h)=\dfrac{(h-1)h^{\frac{1}{h-1}}}{e \log h} \:\:\text{for}\:\: h \neq 1, \text{ and } S(1)=1\:,
\end{equation}
in which $e$ is Euler's number. 

Combining (\ref{FNE1}) and (\ref{FNE2}), we have
\begin{equation}\label{ineq1}
{\displaystyle S(h)^{-c/2} {n \choose {c/2}}}  \Big(\dfrac{\sum_{j=1}^{n}a_j}{n}\Big)^{c/2} \leq \displaystyle \sum_{ 1\leq i_1 < \cdots < i_{c/2} \leq n}a_{i_1} a_{i_2} \cdots a_{i_{c/2}} \leq {\displaystyle {n \choose {c/2}}}  \Big(\dfrac{\sum_{j=1}^{n}a_j}{n}\Big)^{c/2}\:.
\end{equation}
Now, let $a_{i_k}=d_k (d_k-1)$.
%, where $d_k$ is the degree of node $u_k$.
Focusing on the upper bound in (\ref{ineq1}), we then have

\begin{align*}
\displaystyle\sum_{\sigma\subset U \atop |\sigma|=c/2}\prod_{u_i\in \sigma} (d_i)(d_i-1)
& \leq \displaystyle{{n}\choose {c/2}} \Big(\dfrac{\sum_{u_i\in U} (d_i)(d_i-1)}{n}\Big)^{c/2}  \numberthis \label{E102}\\
&\sim \dfrac{n^{c/2}}{(c/2)!} \Big(\dfrac{\sum_{u_i\in U} (d_i)(d_i-1)}{n}\Big)^{c/2}\\
&= \dfrac{1}{(c/2)!} \Big(2\sum_{u_i\in U}  {{d_i}\choose {2}}\Big)^{c/2}. \numberthis \label{E92}
\end{align*}
Similarly, we can establish the following asymptotic lower bound: 
\begin{equation}\label{eq2w}
\displaystyle\dfrac{S(h_u)^{-c/2}}{(c/2)!} \Big(2\sum_{u_i\in U}  {{d_i}\choose {2}}\Big)^{c/2} \leq \displaystyle\sum_{\sigma\subset U \atop |\sigma|=c/2}\prod_{u_i\in \sigma} (d_i)(d_i-1)\:.
\end{equation}

Now, by applying (\ref{E92}) and (\ref{eq2w}) to (\ref{E91}) for both sides of the graph, we obtain the asymptotic upper and lower bounds on $E(N_c)$. In particular, the asymptotic value of the upper bound, used as the approximate value of $E(N_c)$, is calculated as follows:
%following approximation for the expected value of $N_c$:\footnote{It should be noted that in the derivation of (\ref{E92}) and its application to the approximation of $E(N_c)$, we have used the asymptotic value of the upper bound as an approximation in the asymptotic regime.}

\begin{align*}
E(N_c)  &\approx  \dfrac{1}{(c/2)!} \Big(2\sum_{u_i\in U}  {{d_i}\choose {2}}\Big)^{c/2}
              \dfrac{1}{(c/2)!} \Big(2\sum_{w_i\in W}  {{d_i'}\choose {2}}\Big)^{c/2}
              \Big( \dfrac{(\frac{c}{2})!(\frac{c}{2})!}{c} \Big) \times \dfrac{1}{\eta^c}\\
%       &=    \Big(\dfrac{2\sum_{v_i\in V}  {{d_i}\choose {2}}}{1}\Big)^{c/2}
%          \Big(\dfrac{2\sum_{u_i\in U}  {{d_i'}\choose {2}}}{1}\Big)^{c/2}(\dfrac{1}{c})\times \dfrac{1}{m^c}\\
        &= \dfrac{\displaystyle\Big((\frac{2}{\eta}\displaystyle\sum_{i=1}^{n}{{d_i}\choose {2}}) (\frac{2}{\eta}\displaystyle\sum_{i=1}^{m}{{d_i'}\choose {2}})\Big)^{c/2}}{c}.\numberthis \label{E65}
\end{align*}
The asymptotic lower bound on $E(N_c)$ is equal to the upper bound of (\ref{E65}) multiplied by $[S(h_u) \times S(h_w)]^{-c/2}$, proving the claim about the accuracy of the approximation.
This completes the calculation of $E(N_c)$. 

To continue the proof of Theorem~\ref{NewT2}, we need the following lemma, whose proof is provided in Appendix~I.
\begin{lem} 
Consider the ensemble of multigraphs ${\cal G}^*$ in Theorem~\ref{NewT2}, and a fixed multigraph $H$ with more edges than nodes.
% ${\cal G}^*$ is the set of multigraphs with maximum degree $\Delta$, and $\Delta$ is a constant number, 
Then, the expected  number of copies of $H$ in a multigraph in ${\cal G}^*$ is $\mathcal{O}(\frac{1}{n})$.\footnote{The notation $f(x)= \mathcal{O}(g(x))$ is used, if for sufficiently large values of $x$, we have $|f(x)| \leq a |g(x)|$, for some positive value $a$.}
\label{lemxz}
\end{lem}

We now proceed with the calculation of joint factorial moments $E[(N_2)_{r_2} (N_4)_{r_4} \cdots (N_k)_{r_k}]$, where $r_{2i},\:i=1, \ldots, k/2$, are arbitrary non-negative integers, constant with respect to $n$.

\underline{Calculation of joint factorial moments}.
%Next, we calculate the factorial moments. 
We begin with calculating $E[(N_c)_2]$, and show that $E[(N_c)_2] \sim \lambda_{c}^{2}$. This will then be generalized to the asymptotic expression for the joint factorial moments.

We note that $(N_c)_2$ is the number of ordered pairs of two distinct $c$-cycles in  ${\cal G}^*$. 
The two $c$-cycles may or may not intersect. We thus write $(N_c)_2=N'+N''$, where $N'$ is the number of ordered pairs of node disjoint $c$-cycles, and $N''$ is the number of ordered pairs of distinct $c$-cycles that have at least one node in common. Based on Lemma~\ref{lemxz}, we have $E(N'')=\mathcal{O}(\frac{1}{n})$. 
%To see this, observe that each unordered pair of non-node disjoint $c$-cycles corresponds to a copy  of a $H$ that has more edges than nodes. The number of nonisomorphic such $H$�s depends only upon $c$, not on $n$. So The number of them is $\mathcal{O}(1) $. On the other hand, for each such $H$ by Claim 1, the expected number of it is $\mathcal{O}(\frac{1}{n})$. Thus, $E(N'')=\mathcal{O}(\frac{1}{n})$. 
In the following, we prove 
%We now continue by showing that for the expected number of pairs of node-disjoint $c$-cycles in ${\cal G}^*$ we have: 
$E[N']\sim \lambda_{c}^{2}$.

We first count the number ${\cal C}_{kk}$ of ordered pairs of node-disjoint $k$-cycles in a configuration. Similar to the derivation of (\ref{E130}), we have
%
%We can calculate $E[N']$ in the same way as $E[N_c]$. Observe that the ordered pairs of node-disjoint $c$-cycles in ${\cal G}^*$ are in one-to-one correspondence with ordered pairs $$(\{e_1, e_2, \ldots , e_c\}, \{f_1, f_2, \ldots , f_c\})$$ such that each of $\{e_1, e_2, \ldots , e_c\}$ and $\{f_1, f_2, \ldots , f_c\}$ forms a cycle of length $c$ and for each $i,j$, $e_i$ and $f_j$ are not incident (i.e. none of the edges of the first set shares any bin with any of the edges of the second set). Let ${\cal C}_{kk}$ denote the number of ordered pairs of possible node-disjoint $k$-cycles in a configuration that project to disjoint cycles. In order to find two cycles we should choose $k$ bins from $U$ and $k$ bins from $W$. Next, from each bin, we should choose two cells (the order of the two cell is important). So, we have
%\begin{equation}
%\Big( \displaystyle\sum_{\sigma \subset U \atop |\sigma|=k}\prod_{u_i \in \sigma} (d_i)(d_i-1) \Big)\Big( \displaystyle\sum_{\sigma \subset W \atop |\sigma|=k}\prod_{w_i\in \sigma}(d_i')(d_i'-1) \Big)
%\end{equation}
%choices.
%Now, we consider different ordering of the selected $k$ bins from $U$ (i.e. $k!$ options) and  different ordering of the selected $k$ bins from $W$ (i.e. $k!$ options). For each ordered set of $k$ bins, we use the first $k/2$ bins for the first cycle, and other $k/2$ bins for the second cycle. Finally, since by our method  we count each cycle $k$ times and we have two disjoint cycle we should divide the result by $k\times k$. Thus, we have:

\begin{equation} \label{NewE2}
 {\cal C}_{kk} =\Big( \displaystyle\sum_{\sigma \subset U \atop |\sigma|=k}\prod_{u_i \in \sigma} (d_i)(d_i-1) \Big)\Big( \displaystyle\sum_{\sigma \subset W \atop |\sigma|=k}\prod_{w_i\in \sigma}(d_i')(d_i'-1) \Big)\Big( \dfrac{k!k!}{k^2} \Big)\:.
\end{equation}
In order to simplify (\ref{NewE2}), we use the following lemma, whose proof is provided in Appendix I.

\begin{lem}
Let $X=\{ x_1,x_2,\ldots,  x_n\}$ be a set of variables such that for each $i$, $0 < s \leq x_i \leq t$, where both $s$ and $t$ are constant numbers. Also,
let $k$ be a constant number and let $n$ tend to infinity. We then have
\begin{equation} \label{NewEE4}
\Big( \displaystyle\sum_{\sigma \subset X \atop |\sigma|=k}\prod_{x_i \in \sigma} x_i \Big){k \choose k/2}\sim\Big( \displaystyle\sum_{\sigma \subset X \atop |\sigma|=k/2}\prod_{x_i \in \sigma} x_i \Big)^2.
\end{equation}
\label{lemcx}
\end{lem}

%Since $k$ and $\Delta$ are fixed   numbers and $(d_i)(d_i-1)< \Delta^2 $, by  (\ref{NewEE4}) we have:

%\begin{equation} \label{NewEE1}
%\Big( \displaystyle\sum_{\sigma \subset U \atop |\sigma|=k}\prod_{u_i \in \sigma} (d_i)(d_i-1) \Big){k \choose k/2}\sim\Big( \displaystyle\sum_{\sigma \subset U \atop |\sigma|=k/2}\prod_{u_i \in \sigma} (d_i)(d_i-1) \Big)^2.
%\end{equation}
By the application of Lemma~\ref{lemcx} to (\ref{NewE2}), we obtain the following:

\begin{equation} \label{NewEE2}
 {\cal C}_{kk} \sim\Big( \displaystyle\sum_{\sigma \subset U \atop |\sigma|=k/2}\prod_{u_i \in \sigma} (d_i)(d_i-1) \Big)^2\Big( \displaystyle\sum_{\sigma \subset W \atop |\sigma|=k/2}\prod_{w_i\in \sigma}(d_i')(d_i'-1) \Big)^2\Big( \dfrac{(k/2)!(k/2)!}{k} \Big)^2\:.
\end{equation}

This together with (\ref{E130}) result in

\begin{equation} \label{NewEE3}
 {\cal C}_{kk} \sim {\cal C}_{k}^2.
\end{equation}

%The probability that a given pair of disjoint $c$-cycles is contained in a random configuration is $\dfrac{ (\eta-2c)!}{\eta!}$, Thus, by (\ref{NewEE3}) and (\ref{E91}) 
We thus have

\begin{align*}
E[N'] &= \dfrac{{\cal C}_{cc} \times(\eta-2c)!}{\eta!}\\
      &\sim {\cal C}_{c}^2 \times \dfrac{1}{\eta^{2c}}\\
        &\sim  \lambda_c^2\:,
\end{align*}
where in the last step, we have used (\ref{E91}).

%Consequently,
%$$ E(N_c)_2=E[N']+E[N'']\sim \lambda_c^2.$$

We note that the joint factorial moment under consideration is the expected value of the product of the number of ordered $r_i$ distinct  $i$-cycles, for even values $i=2,\ldots,k$. This can be interpreted as the expected number 
of sequences of $\beta=r_2+r_4+\ldots+r_k$ distinct cycles such that the first $r_2$ have length $2$, the next $r_4$ have length $4$, and so on. In the following, we call such sequences $\beta$-sequences.
%Let $m$ be a fixed even number. We can apply the same argument for any fixed $r_2,\ldots, r_m \in \mathbb{N} \cup \{0\}$, we  consider the number of sequences of $r_2+r_4+\ldots+r_m$ distinct cycles such that the first $r_2$ have length $2$, the next $r_4$ have length $4$, etc. 
Similar to the approach for two cycles, the number of $\beta$-sequences can be written as $N' + N''$, where $N'$ counts the sequences of node-disjoint cycles and $N''$ counts the sequences of distinct cycles such that in each sequence there are at least two cycles that share at least one node. For $N''$, by Lemma~\ref{lemxz}, we have $E(N'')=\mathcal{O}(\frac{1}{n})$.

Now, we study $N'$. We use the notation ${\cal C}'$ to denote the number of possible $\beta$-sequences with node-disjoint cycles in a configuration.  
Similar to (\ref{NewE2}), we have
\begin{equation}\label{NewEE5}
 {\cal C}' =\Big( \displaystyle\sum_{\sigma \subset U \atop |\sigma|=\alpha}\prod_{u_i \in \sigma} (d_i)(d_i-1) \Big)\Big( \displaystyle\sum_{\sigma \subset W \atop |\sigma|=\alpha}\prod_{w_i\in \sigma}(d_i')(d_i'-1) \Big)\Big( \dfrac{\alpha!\alpha!}{2^{r_2}4^{r_4}\ldots k^{r_k} } \Big)\:,
\end{equation}
where $\alpha =r_2+2r_4+\ldots+(k/2) r_k$.

Lemma~\ref{lemcx} can be extended to the following asymptotic equality:
%Let $X=\{ x_1,x_2,\ldots,  x_n\}$ be a set of variables such that for each $i$, $0\leq x_i \leq c$, where $c$ is a constant number and $n$ tends to infinity. Also let $\alpha$ be a constant number. We have the following (note that $\alpha =\frac{2}{2}\times r_2+\frac{4}{2}\times r_4+\ldots+\frac{m}{2}\times r_m$).
\begin{equation}\label{NewEE6}
\Big( \displaystyle\sum_{\sigma \subset X \atop |\sigma|=\alpha}\prod_{x_i \in \sigma} x_i \Big){\alpha \choose \frac{2}{2},\ldots,\frac{4}{2},\ldots,\frac{k}{2}}\sim
\Big( \displaystyle\sum_{\sigma \subset X \atop |\sigma|=2/2}\prod_{x_i \in \sigma} x_i \Big)^{r_2}
\Big( \displaystyle\sum_{\sigma \subset X \atop |\sigma|=4/2}\prod_{x_i \in \sigma} x_i \Big)^{r_4}\ldots
\Big( \displaystyle\sum_{\sigma \subset X \atop |\sigma|=k/2}\prod_{x_i \in \sigma} x_i \Big)^{r_k}
\end{equation}
Using (\ref{NewEE6}) in (\ref{NewEE5}), we obtain
\begin{align*}
 {\cal C}' &\sim
\Big( \displaystyle\sum_{\sigma \subset U \atop |\sigma|=2/2}\prod_{u_i \in \sigma} (d_i)(d_i-1) \Big)^{r_2}
\Big( \displaystyle\sum_{\sigma \subset U \atop |\sigma|=4/2}\prod_{u_i \in \sigma} (d_i)(d_i-1) \Big)^{r_4}\ldots
\Big( \displaystyle\sum_{\sigma \subset U \atop |\sigma|=k/2}\prod_{u_i \in \sigma} (d_i)(d_i-1) \Big)^{r_k}\\
 &\times \Big( \displaystyle\sum_{\sigma \subset W \atop |\sigma|=2/2}\prod_{w_i \in \sigma} (d_i')(d_i'-1) \Big)^{r_2}
\Big( \displaystyle\sum_{\sigma \subset W \atop |\sigma|=4/2}\prod_{w_i \in \sigma} (d_i')(d_i'-1) \Big)^{r_4}\ldots
\Big( \displaystyle\sum_{\sigma \subset W \atop |\sigma|=k/2}\prod_{w_i \in \sigma} (d_i')(d_i'-1) \Big)^{r_k}\\
 &\times \dfrac{\Big(((2/2)!)^{r_2}((4/2)!)^{r_4}\ldots ((k/2)!)^{r_k}  \Big)^2}{2^{r_2}4^{r_4}\ldots k^{r_k} } \:.
\end{align*}
By the application of (\ref{E130}) and (\ref{E91}) to the expected value of the above equation, we then have $E[N']\sim \lambda_2^{r_2}\ldots \lambda_k^{r_k}$. Hence,
\begin{center}
$\displaystyle E[(N_2)_{r_2}(N_4)_{r_4}\ldots (N_{k})_{r_k}]\rightarrow \prod_{i=1}^{k/2}\lambda_{2i}^{r_{2i}}$\,\,\, as $n\rightarrow \infty$.
\end{center}

Thus, by Theorem \ref{NewT1}, for any fixed even value of $k \geq 2$, the random variables $ N_2, \ldots, N_k$,
are asymptotically independent Poisson random variables with $N_c$ having the expected value
$$
E(N_c)\approx  \dfrac{\displaystyle\Big((\frac{2}{\eta}\displaystyle\sum_{i=1}^{n}{{d_i}\choose {2}}) (\frac{2}{\eta}\displaystyle\sum_{i=1}^{m}{{d_i'}\choose {2}})\Big)^{c/2}}{c}\:.
$$
}\end{proof}

{\bf Proof of Theorem~\ref{T01}.} 
%Now, we return to the process of constructing random bipartite graphs (the ensemble ${\cal G}$). In the process of finding images of configurations (construction of the ensemble ${\cal G}^*$), we can discard bipartite graphs with parallel edges. The set of corresponding configurations (with no parallel edges) is denoted by $\Omega$. We note that a uniform distribution on the configurations in $\Omega$ induces a uniform distribution on ${\cal G}$. This is because corresponding to each graph in ${\cal G}$, we have the same number $d_1! \times \cdots \times d_n! \times d'_1! \times \cdots \times d'_m$ of configurations in $\Omega$.
We note that multigraphs in ${\cal G}^*$ are simple
%\footnote{Recall that a bipartite multigraph is said to be simple if it has no parallel edges.} 
if and only if $N_2=0$, and also that  ${\cal G}^*$ conditioned on $N_2=0$ yields ${\cal G}$. Let ${\cal S}$ denote the event that the multigraphs in ${\cal G}^*$ are simple. By Theorem \ref{NewT2}, we have $\Pr({\cal S}) \sim e^{-\lambda_2}$,
and thus,  $\Pr({\cal S})> 0$. We now show that any property $P$ that holds true asymptotically almost surely (a.a.s.) for ${\cal G}^*$ (including that of Theorem~\ref{NewT2} on cycle distributions), also holds true a.a.s. for ${\cal G}$. 
Let $\mathcal{P}^*$ and $\mathcal{P}$ denote the events 
that multigraphs in ${\cal G}^*$ and bipartite graphs in ${\cal G}$ have Property $P$, respectively. We then have 
%Thus, we have the following useful claim.\\
%{\bf Claim 3.}
%If the property $\mathcal{P}$ holds asymptotically almost surely\footnote{We say that an event $E_n$, describing a property of a random structure depending on a parameter $n$, holds asymptotically almost surely if $\Pr(E_n)\rightarrow 1$ as $ n \rightarrow \infty$.} (a.a.s.) for ${\cal G}^*$, then the property $\mathcal{P}$ holds a.a.s. for ${\cal G} $.\\ \\
%{\bf Proof of Claim 3.}
%\begin{align*}
\begin{equation}
 \Pr (\bar{\mathcal{P}})=\Pr (\bar{\mathcal{P}^*} | {\cal S}) = \dfrac{\Pr(\bar{\mathcal{P}^*} \cap {\cal S})}{\Pr({\cal S})} 
                                                   \leq \dfrac{\Pr(\bar{\mathcal{P}^*})}{\Pr({\cal S})}\rightarrow 0\:,
\end{equation}
where the last part follows from the fact that Property $P$ holds true a.a.s. on ${\cal G}^*$. $\blacksquare$
%\end{align*}
%This completes the proof of claim.  

%We note that convergence in distribution is really a property of the distributions of the random variables and does not require the variables to be defined on the same probability space (see \cite{MR1782847}, pp 9). Consequently, by Theorem \ref{NewT2} and Claim 3, we have the following theorem for random bipartite graphs (here the property is: For any fixed even value of $k \geq 4$, the random variables $ N_4, N_6, \ldots, N_k$, are asymptotically independent Poisson random variables with $N_c$ having the expected value $\lambda_c$).

\begin{rem}
For bi-regular graphs, where $\Delta_u=\delta_u$ and $\Delta_w=\delta_w$, we have $S(h_u)=S(h_w)=1$, and thus
the asymptotic upper and lower bounds  on $E(N_c)$, derived in Theorem~\ref{T01}, coincide. In this case, the asymptotic approximation provided for $E(N_c)$ in Theorem~\ref{T01} turns into an asymptotic equality. 

For irregular graphs, we have $h_u > 1$ or $h_w > 1$, and thus $[S(h_u) \times S(h_w)]^{-c/2} < 1$. Specht's ratio $S(h)$ is a monotone increasing function on $(1,\infty)$, and thus the asymptotic lower bound on $E(N_c)$ decreases monotonically with increase in $h_u$ or $h_w$. The numerical results in Section~\ref{sec4}, however, show that the asymptotic upper bound presented in Theorem~\ref{T01} is often much tighter than the asymptotic lower bound presented in this theorem.
%The rate of increase of $S(h)$ with increasing $h$ is however, rather slow. For example, we have $S(2)=1.061$, $S(5)=1.367$ and $S(10)=1.857$. This implies that for practical values of $h_u$ and $h_w$, the gap between the asymptotic lower and upper bounds are rather small. 
%This will be further discussed in Section~\ref{sec4}.  
%This results in a more accurate approximation of $E(N_c)$ for bi-regular graphs in comparison with irregular graphs. 
%In (\ref{E102}), equality holds if and only if $d_1=d_2=\ldots=d_n$. Thus, in the case of regular degree distributions, approximation (\ref{E92}) is more accurate in comparison with the case of irregular degree distributions, where (\ref{E92}) is an approximate upper bound.
\label{Rx}
\end{rem}

\begin{cor} \label{C1}
Let $G=(U\cup W,E)$ be a random bi-regular graph in which all the nodes in $U$ have the same degree $d_u$ and all the nodes in $W$ have the same degree $d_w$.
Consider the ensemble of such graphs as the number of nodes tends to infinity. In this case, for a fixed even value $k$,
random variables $N_4(G),\ldots,N_k(G)$, are independent with Poisson distribution, where the expected value of $N_c$ is given by
\begin{equation}
E(N_c)\sim  \dfrac{\Big((d_u-1)(d_w-1)\Big)^{c/2}}{c}\:.
\end{equation}
\end{cor}

\subsection{An Application}
\label{subsec2.2}

It is well-known that the performance of LDPC codes in the error floor region is determined by certain substructures of the code's Tanner graph, referred to as trapping sets~\cite{richardson2003error}. The error floor performance of an LDPC code is not only a function of the Tanner graph of the code but also depends on the channel model, quantization scheme and the iterative algorithm used for the decoding. Depending on the scenario, different categories (types) of trapping sets may prove to be relevant. Such categories include elementary trapping sets (ETS), leafless ETSs (LETS), absorbing sets, and stopping sets. For example, while stopping sets are known to be the culprit in belief propagation decoding of LDPC codes over the binary erasure channel (BEC)~\cite{Di}, LETSs are the relevant structures in the context of iterative decoding of LDPC codes over the additive white Gaussian noise (AWGN) channel~\cite{HB-IT1},~\cite{Y-Arxiv}.

Very recently, it was shown in~\cite{IT2cycle} that, regardless of the type of a trapping set structure, its asymptotic average multiplicity in a random ensemble of Tanner graphs depends only on the trapping set's constituent cycles. In particular, a structure with no cycle, with only one cycle, and with more than one cycle has an asymptotic average multiplicity of infinity, a non-zero constant, and zero, respectively.  For the non-trivial case where the structure has only a single (chordless) cycle, the asymptotic average multiplicity can be estimated using Theorem~\ref{T01}.

\begin{ex}
Consider the random ensemble of bi-regular LDPC codes with $d_u =3$ and $d_w=6$. The error floor of this ensemble over the AWGN channel is determined by the distribution of LETS structures of the codes. It is proved in \cite{IT2cycle} that among all $(a,b)$ classes\footnote{A trapping set is often identified by the number of its variable nodes $a$, and the number of unsatisfied check nodes $b$ in its induced subgraph. Such a trapping set is said to belong to the class of $(a,b)$ trapping sets.} of LETS structures for this ensemble, only those with $b=a$ have a non-zero asymptotic multiplicity. Such structures correspond to chordless cycles of length $2a$. We can thus use Corollary~\ref{C1} and estimate the average number of such structures by $10^a/(2a)$. One should note that while Theorem~\ref{T01} and Corollary~\ref{C1} consider both chordless cycles and cycles with chords, the multiplicity of cycles with chords tends to zero asymptotically~\cite{IT2cycle}, and thus the results given here provide good estimates for the asymptotic multiplicity of chordless cycles.

To demonstrate the accuracy of this estimate at finite block lengths, we have randomly constructed five LDPC codes with $d_u=3$ and $d_w=6$, and block lengths $n = 816, 1008, 4000, 20000,$ and $50000$. All the codes have girth $6$. The multiplicity of $(a,a)$ LETS structures of these codes for $a=3, 4, 5$, are listed in Table~\ref{NewTable}, along with the estimate of Corollary~\ref{C1}.  As can be seen from Table~\ref{NewTable}, the estimates for different values of $a$ are rather accurate even for relatively short block lengths.
\end{ex}

\begin{table}[ht]
\caption{Multiplicities of $(3,3)$, $(4,4)$, and $(5,5)$ LETSs  for randomly constructed bi-regular LDPC codes with $d_u=3$ and $d_w=6$, in comparison with the asymptotic expected values of Corollary~\ref{C1}}
\begin{center}
\scalebox{1}{
\begin{tabular}{ |c||c|c|c|c|c||c|  }

\hline
$(a,a)$ Class &   \multicolumn{ 5}{  c| }{  Block Length} & Expected value by  \\
\cline{2-6}
&  816 & 1008 & 4000 &  20000 &  50000   & Corollary~\ref{C1}\\ \hline \hline

(3,3) &    132    &   165     &  171    & 161    & 178   & 166 \\
\hline
(4,4) &    1491   &  1252     & 1219    & 1260   & 1268  & 1250 \\
\hline
(5,5) &    9169   &  10019    & 9935    & 10046  & 10231 & 10000\\
\hline
\end{tabular}
}
\end{center}
\label{NewTable}
\end{table}

\section{Random Lifts of an Arbitrary Bipartite Base Graph}
\label{secN1}

In this section, we study the cycle distribution of protograph-based LDPC codes that are random lifts of a base graph with no parallel edges.
%Here, we calculate the average number of cycles of a given length
%for the random lifts of the fully-connected bipartite graphs and bi-regular   graphs.
The following result shows that, similar to random bipartite graphs, for random lifts also, the cycles of different length have independent Poisson distributions.
%Fortin and Rudinsky in \cite{fortin2012asymptotic} proved the following important result about the distribution of cycles in  the random lifts of the protographs.

\begin{theo}\label{TofFortin} \cite{fortin2012asymptotic}
For a random $N$-lift of a protograph $G$,  as $N$ tends to infinity, the distributions of cycles of different length $c$ tend to independent Poisson distributions
with the expected value equal to $T(G,c)$, where $T(G,c)$ is the number of TBC walks of length $c$ in $G$.
\end{theo}

In the following, we calculate $T(G,c)$ for two special cases of base graphs commonly used in the construction of protograph-based LDPC codes: fully-connected and bi-regular.
Although, fully-connected graphs are themselves a special case of bi-regular graphs, in the following, we first consider the case of fully-connected graphs, since for this case, we can in fact, derive an exact expression for $T(G,c)$.
For the more general case of bi-regular base graphs, our approximation is in the form of an upper bound.
%that the asymptotic expected values of the number of cycles of a give length for random lifts of the protographs of a fully-connected base graph and those
%of a bi-regular random bipartite graph of similar degree distribution are identical.

\subsection{Calculation of $T(G,c)$ for fully-connected base graphs}

%Here, we calculate $T(G,c)$ for fully-connected bipartite graphs.
\begin{theo}\label{T3}
Let $G=(U \cup W)$ be a  fully-connected bipartite graph  with $|U|=a$ and $|W|=b$. For any even value $c \geq 4$, we have
$$
T(G,c)= \dfrac{(a-1)(b-1)}{c} \Big( (-1)^{c/2} + (a-1)^{c/2-1}\Big) \Big( (-1)^{c/2} + (b-1)^{c/2-1}\Big)\:.
$$
\end{theo}

\begin{proof}{
To calculate $T(G,c)$, we consider the number of TBC walks of length $c$, $R_{c,e}$, that go through a specific edge $e$ in the base graph $G$.
Due to the symmetry of $G$, this number is independent of $e$. In the rest of the proof, we thus use the notation $R_c$ for this number. Since there are $ab$ edges in $G$, we have
\begin{equation}\label{NEW1}
{
T(G,c)=    \dfrac{ab \times R_c }{c}\:,
}
\end{equation}
where the division by $c$ is because each TBC walk is accounted for $c$ times through its $c$ edges.

To calculate $R_c$, we note that any TBC walk of length $c$ in the fully-connected based graph can be uniquely described by two interleaving sequences of
variable and check nodes, where each sequence corresponds to a closed walk of length $c/2$ in the complete graph $K_a$ and $K_b$, respectively. Suppose
that the number of closed walks of length $c/2$ starting from a specific node in $K_a$ is denoted by ${\cal W}^a_{c/2}$. We thus have
\begin{equation}
R_c = {\cal W}^a_{c/2} \times {\cal W}^b_{c/2}\:.
\label{eqjg}
\end{equation}

To obtain ${\cal W}^a_k$, we need to calculate a diagonal element of $A_a^k$, where $A_a$ is the $a \times a$ adjacency matrix of $K_a$.
It is easy to see that the $k$-th power of $A_a$ has the following general form
%Assume that  the statement holds for each $k\leq i$ and
$$A^k_a =
 \begin{pmatrix}
  \alpha_k    & \beta_k    & \cdots & \beta_k \\
  \beta_k    & \alpha_k    & \cdots & \beta_k \\
  \vdots     & \vdots       & \ddots & \vdots  \\
  \beta_k    & \beta_k      & \cdots & \alpha_k
 \end{pmatrix},$$
where
\begin{align*}
\alpha_1   &= 0, \,\,\, \alpha_{k+1}=(a-1)\beta_k\\
\beta_1    &= 1, \,\,\,  \beta_{k+1}= \alpha_k+ (a-2)\beta_k = (a-2)\beta_k + (a-1)\beta_{k-1}.
\end{align*}

To solve the recursion $\beta_{k+1}=  (a-2)\beta_k + (a-1)\beta_{k-1}$, we solve the corresponding quadratic equation $x^2-(a-2)x-(a-1)=0$.
The roots of this equation are $-1$ and $a-1$. Thus, $\beta_k= \gamma (-1)^k+ \gamma'(a-1)^k$. Using $\beta_1=1$ and $\beta_2=a-2$, we obtain $\gamma'= - \gamma=\frac{1}{a}$.
Hence, $\beta_k=\frac{-1}{a} (-1)^k + \frac{1}{a}(a-1)^k$. We thus have

\begin{equation}\label{E2}
{\cal W}^a_k = \alpha_k = (a-1) \beta_{k-1} = \frac{a-1}{a} (-1)^k + \frac{1}{a}(a-1)^k\:.
\end{equation}
%\\\\
Combining (\ref{E2}) with (\ref{eqjg}) and (\ref{NEW1}) completes the proof.
}\end{proof}

\begin{cor}\label{C2}
Let $G=(U \cup W)$ be a  fully-connected bipartite graph  with $|U|=a$ and $|W|=b$. For any even value $c \geq 4$, we have
$$
T(G,c)  \approx \dfrac{\Big((a-1)(b-1)\Big)^{c/2}}{c}.$$
\end{cor}

\begin{rem}{
Combination of Theorem \ref{TofFortin} and Corollary \ref{C2}, and the comparison with the result of Corollary \ref{C1} show that, in the asymptotic regime, where the size of the graph tends to infinity, the cycle distributions of random lifts of a fully-connected base graph are identical to those of random bi-regular graphs with the same variable and check node degrees.
% Let $\tilde{G}$ be the random lift  of a protograph, where the base graph is a fully-connected bipartite graph $G=(U \cup W)$ with $|U|=a$ and $|W|=b$. By Theorem \ref{TofFortin} and Corollary \ref{C2}  the expected value of the number of cycles of length $c$ is at most $\frac{[(a-1)(b-1)]^{c/2}}{c}$. On the other hand, let $G=(U\cup W,E)$ be a random bi-regular graph in which all the nodes in $U$ have the same degree $d_u$ and all the nodes in $W$ have the same degree $d_w$. Then by Corollary \ref{C1}  the expected value of the number of cycles of length $c$ is equal to $\frac{[(d_u-1)(d_w-1)]^{c/2}}{c}$. Thus, they have the same number of cycles for each length.
}\label{R32}\end{rem}

\subsection{Calculation of $T(G,c)$ for  general bi-regular  graphs}

In this part, we consider the graphs that are bi-regular but not necessarily fully-connected. %Assume that $G$ is bi-regular and without loss of generality assume that   $d_u = g$ and $d_w = h$.

\begin{theo}\label{T4}
Let $G=(U \cup W)$ be a bi-regular graph.
% with node degrees equal to $d_u$ and $d_w$ on the two sides of the graph, respectively.
Then, for any even value $c \geq 4$, we have
$$
T(G,c)  \leq \dfrac{|U|d_u}{c}\Big((d_u-1)(d_w-1)\Big)^{c/2-1}.$$
\end{theo}

\begin{proof}{
By counting the TBC walks in $G$ from the viewpoint of the edges, we have
\begin{equation}
T(G,c) \leq \frac{|U| d_u}{c} K_c\:,
\label{eqst}
\end{equation}
where $K_c$ is defined as the maximum number of TBC walks of length $c$ to go through a specific edge in $G$ (the maximum is taken over all the edges in $G$).
%Define
%\\ \\
%$K_c:=\max_{e\in E(G)} \{ \text{The number of TBC walks of length } c \text{  that go through a specific edge }e \}$.
%\\ \\
%
%\begin{align*}
%E(X_c)  &\simeq \dfrac{\displaystyle \sum \Pr(I_{(\alpha,\beta, W, e)}=1)}{c}\\
%        & \leq \dfrac{N^2 (Ug) \times K_c \times \frac{1}{N^2} }{c}.
%\end{align*}
Now, for a given edge $e$ in $G$, consider a potential TBC walk in $G$ that starts from $e$. There are $[(d_u-1)(d_w-1)]^{c/2-1}$ possibilities for selecting the following $c-2$ edges
of such a potential TBC walk. For the last edge of the TBC walk, there would be only one choice $e'$ that can connect the end node of the last edge to the beginning node of $e$. This is
if such an edge $e' \neq e$ exists. We thus have $K_c \leq [(d_u-1)(d_w-1)]^{c/2-1}$. This together with (\ref{eqst}) completes the proof.
}\end{proof}

\begin{rem}{\label{R1}
Note that for a fully-connected base graph, the upper bound of Theorem~\ref{T4} is approximately equal to the value given in Corollary~\ref{C2}.
%Let $\tilde{G}$ be a random lift  of a protograph, where the base graph is a bi-regular  graph $G=(U \cup W)$, with node degrees equal to $d_u$ and $d_w$ on the two sides of the graph, respectively. By Theorem \ref{TofFortin} and Theorem \ref{T4} the expected value of the number of cycles of length $c$ is at most $\frac{|U|d_u}{c}[(d_u-1)(d_w-1)]^{c/2-1}$. Since the base graph is a small graph we have $|U|d_u \simeq(d_u-1)(d_w-1)$. Thus, the expected value of the number of cycles of length $c$ is at most $\frac{[(d_u-1)(d_w-1)]^{c/2}}{c}$. On the other hand, let $G=(U\cup W,E)$ be a random bi-regular graph in which all the nodes in $U$ have the same degree $d_u$ and all the nodes in $W$ have the same degree $d_w$. Then by Corollary \ref{C1} the expected value of the number of cycles of length $c$ is equal to $\frac{[(d_u-1)(d_w-1)]^{c/2}}{c}$. Thus, they have the same number of cycles for each length.
}\end{rem}

\section{Random Cyclic Lifts of an Arbitrary Bipartite Base Graph}
\label{sec3}
In this section, we focus on random cyclic liftings of degree $N$ of a given bipartite base graph $G$. The randomness is with respect to the exponent matrix $P$, where
each non-infinity element of $P$ is selected in an independent and identically distributed (i.i.d.) fashion from a uniform distribution over $\mathbb{Z}_N$. In the following, we first
derive upper and lower bounds on the expected value of the number of $c$-cycles, followed by an upper bound on the variance.

\subsection{Calculation of $E(N_c)$}
We use the notation ${\cal T}(G,c)$ to denote the set of all TBC walks of length $c$ in a base graph $G$. This set has size $T(G,c)$. To derive our results, we need to partition ${\cal T}(G,c)$ into
three subsets ${\cal T}_1(G,c), {\cal T}_2(G,c)$, and ${\cal T}_3(G,c)$. The partition ${\cal T}_1(G,c)$ is the set of all prime ZP TBC walks of length $c$ in $G$, while
${\cal T}_2(G,c)$ consists of all TBC walks $w$ of length $c$ in $G$ such that $w$ contains at least a ZP TBC subwalk. The partition ${\cal T}_3(G,c)$
%is the set of TBC walks of length $c$ in $G$ that are neither in ${\cal T}_1(G,c)$, nor in ${\cal T}_2(G,c)$, and are a multiple of one of their TBC subwalks. The partition ${\cal T}_4(G,c)$
covers the rest of the TBC walks of length $c$ in $G$, i.e.,
${\cal T}_3(G,c) = {\cal T}(G,c) \setminus ({\cal T}_1(G,c) \cup {\cal T}_2(G,c))$. In the following, for simplicity of notations, we use ${\cal T}$ for ${\cal T}(G,c)$, and ${\cal T}_i$ for ${\cal T}_i(G,c)$.
%\cup_{i=1}^{3}{\cal T}_i(G,c)$

Consider an edge $e$ involved in a TBC walk $w$ in ${\cal T}$. Assume that $e$ is traversed $i$ times in one direction and $j$ times in the opposite direction. The contribution of $e$ in ${\cal P}(w)$ is thus $(i-j)p_e$, where $p_e$ is the permutation shift of $e$. In this case, we say edge $e$ is of {\em multiplicity} $|i-j|$ in $w$. We now organize the contribution of different edges of $w$ in ${\cal P}(w)$ in accordance with their multiplicity, as follows:
\begin{equation}
{\cal P}(w) = \sum_{e \in E_1} p_e + 2 \times \sum_{e \in E_2} p_e + \cdots + k \times  \sum_{e \in E_k} p_e\:,
\label{eq38}
\end{equation}
where $E_i$ is the set of edges of multiplicity $i$, and $k$ is the largest multiplicity of edges in $w$. In (\ref{eq38}), with a slight abuse of notation, we have used $p_e$ to denote either $p_e$ or $-p_e$ depending on the sign of $i-j$. In relation to (\ref{eq38}), we say TBC walk $w$ is of {\em degree} $k$. Assuming that $\ell$ summations (out of $k$) in (\ref{eq38}) are non-zero, we refer to $w$ as a TBC walk of {\em weight} $\ell$. Clearly, ZP TBC walks have both degree
zero and weight zero.
%Assume that the total number TBC walks of length $c$ in the base graph $G$ is $T(G,c)$. Here, we show that  asymptotically, the cycle distribution is independent
%of the lifting degree $N$ and only depends on the structure of the base graph.
\begin{lem}
Consider a random cyclic $N$-lift of a base bipartite graph $G$ with no parallel edges, and consider a TBC walk $w$ of length $c$ and weight $\ell \geq 1$ in $G$. We then have
%Let ${\Pr}^{(i)}({\cal P}(w)=0)$ denote the probability that ${\cal P}(w)$, the permutation shift corresponding to $w$, is zero, assuming that $w$ is in ${\cal T}_i(G,c)$, for $i = 1, 2, 3, 4$. We then have
%\begin{equation}
%{\Pr}^{(4)}({\cal P}(w)=0) = \frac{1}{N}\:,
%\label{eqak}
%\end{equation}
%and
\begin{equation}
\frac{1}{N^{\ell}} \leq {\Pr}({\cal P}(w)=0) \leq \frac{c}{4N}\:.
\label{eqsj}
\end{equation}
%where ${\cal P}(w)$ is the permutation shift corresponding to $w$.
\label{lem25}
\end{lem}
\begin{proof}{
We first note that the degree $k$ of a TBC walk of length $c$ is at most $c/4$. This can be easily seen by noting that passing through an edge $e$, $k$ times, requires passing through $k$ closed walks, each containing $e$. Since graph $G$ is assumed to have no parallel edges and is bipartite, the length of each such closed walk is at least $4$.

For each non-empty set $E_i$, the corresponding summation in (\ref{eq38}), denoted by $X_i$, takes one of the $N$ values in $\mathbb{Z}_N$ with equal probability. Also, different summations in (\ref{eq38}) are independent, since they share no permutation shifts. The relationship (\ref{eq38}) is then a linear integer combination of i.i.d. random variables $X_i$'s, and we are interested in evaluating the probability that this linear combination is equal to zero modulo $N$. Considering that the weight of $w$ is $\ell$, we are thus interested in the probability that the following equation is satisfied:
\begin{equation}
j_1 X_{j_1} + \cdots + j_{\ell} X_{j_{\ell}} = 0 \mod N\:,
\label{eq21}
\end{equation}
where $j_i, i = 1, \ldots, \ell$, are the indices corresponding to non-zero random variables. The lower bound of (\ref{eqsj}) immediately follows by noticing that setting all the random variables equal to zero satisfies (\ref{eq21}).

For the upper bound, consider (\ref{eq21}), in which all the random variables except $X_{j_i}$ are fixed. The number of solutions to this equation (considering $X_{j_i}$ as the variable) is then at most $\gcd(j_i,N)$, where $\gcd(\cdot,\cdot)$ denotes the greatest common divisor. This implies that the probability of
(\ref{eq21}) being satisfied is upper bounded by $\gcd(j_i,N)/N$, and thus by $\min\{\gcd(j_1,N)/N,\cdots,\gcd(j_{\ell},N)/N\}$. Now, the upper bound in (\ref{eqsj}) follows from $\gcd(j_i,N) \leq j_i \leq k \leq c/4$, for any $j_i$.
}\end{proof}

\begin{lem}
Consider a random cyclic $N$-lift of a base bipartite graph $G$ with no parallel edges, and consider a TBC walk $w$ in ${\cal T}_1$. We then have
\begin{equation}
\Pr(A_w) \geq 1 - \frac{c^3}{4N} \:,
\label{eqhu}
\end{equation}
where $A_w$ is the event that none of the subsequences of permutation shifts for $w$ corresponds to a TBC walk with zero permutation shift.
\label{lemtg}
\end{lem}
\begin{proof}{
%We note that the TBC walks in ${\cal T}_1(G,c)$ or ${\cal T}_4(G,c)$ are not a multiple of another TBC walk.
Denote by $\bar{A_w}$, the complement event of $A_w$. It is easy to see that the number of subwalks of $w$ is upper bounded by $c^2$. Each such subwalk, based on Lemma~\ref{lem25}, is a TBC walk of permutation zero with probability at most $c/(4N)$. We thus have $\Pr(\bar{A_w}) \leq c^3/(4N)$. This together with
$\Pr(A_w) = 1- \Pr(\bar{A_w})$, completes the proof.
}\end{proof}

%\begin{lem}
%Consider a random cyclic $N$-lifting of a base bipartite graph $G$ with no parallel edges, and consider a TBC walk $w$ in ${\cal T}_3(G,c)$.
%We then have
%\begin{equation}
%\Pr(A_w | {\cal P}(w)=0) \leq 1- \frac{4}{c} \:,
%\label{eqlo}
%\end{equation}
%where $A_w$ is defined as in Lemma~\ref{lemtg}.
%\label{lembg}
%\end{lem}
%\begin{proof}{Consider a TBC walk $w$ in ${\cal T}_3(G,c)$ with the shortest TBC subwalk $w'$. Let ${\cal P}(w)=\pm k\times {\cal P}(w')$, for some $k \geq 2$. Given ${\cal P}(w)=0$, the equation ${\cal P}(w)=\pm k\times {\cal P}(w')=0 \mod N$,  has $t=gcd(k,N)$ solutions for ${\cal P}(w')$, including ${\cal P}(w') =0$. Due to the random nature of lifting, each solution occurs with equal probability. We thus have $\Pr(\bar{A_w} | {\cal P}(w)=0)) \geq 1/t$. This together with $t \leq k \leq c/4$, and $\Pr(A_w | {\cal P}(w)=0) = 1- \Pr(\bar{A_w} | {\cal P}(w)=0)$ completes the proof.
%}\end{proof}

\begin{theo}\label{Th2}
Let  $\tilde{G}$ be a random cyclic $N$-lift of a base bipartite graph $G$ with no parallel edges. For any even value $c \geq 4$, we have
$$
 (N - \frac{c^3}{4}) \times T_1 \leq E[N_c(\tilde{G})] \leq N\times T_1 + \frac{c}{4} \times T_3 \:,
$$
where $T_i$ is the size of the set ${\cal T}_i$.
\end{theo}

\begin{proof}{
%Assume that $c$ is a fixed number.
Consider the base graph $G$, and the ensemble of random cyclic $N$-lifts $\tilde{G}$.
%Denote by ${\cal T}(G,c)$, the set of all TBC walks of length $c$ in $G$.
The number of cycles of length $c$ in $\tilde{G}$ is then given by the following random variable:
\begin{equation}\label{EN1}
N_c(\tilde{G}) = N \sum_{w \in {\cal T}} I(\{{\cal P}(w)=0\} \cap A_w),
\end{equation}
where $I(\cdot)$ is the indicator function, and $A_w$ is the event as defined in Lemma~\ref{lemtg}.
By (\ref{EN1}), and using the definition of conditional probability, we have
%\begin{equation}
%E[N_c(\tilde{G})] = N \sum_{w \in {\cal T}(G,c)}  \Pr(\{{\cal P}(w)=0\} \cap A_w)\})\:.
%\label{EM1}
%\end{equation}
%
%Using the definition of conditional probability, from (\ref{EM1}), we have
\begin{equation}
E[N_c(\tilde{G})] = N \sum_{w \in {\cal T}}  \Pr(\{{\cal P}(w)=0\} \cap A_w)\}) = N \sum_{w \in {\cal T}} \Pr({\cal P}(w)=0) \times \Pr(A_w | {\cal P}(w)=0).
\label{EC1}
\end{equation}
Consider breaking down the summation over the set ${\cal T}$ in (\ref{EC1}) to summations over the three partitions of ${\cal T}$, i.e., over the sets ${\cal T}_1$, ${\cal T}_2$, and ${\cal T}_3$.
In the following, we evaluate the probability $\Pr(\{{\cal P}(w)=0\} \cap A_w)$, for
TBC walks $w$ in the three sets, respectively.

For each TBC walk $w$ in ${\cal T}_1$, by the definition of ${\cal T}_1$, we have: $\Pr(\{{\cal P}(w)=0\})=1$.
Combining this with Lemma~\ref{lemtg}, we have
%$ \Pr(A_w | {\cal P}(w)=0)=\Pr(A_w)$.
%Also, $\Pr(A_w) = 1 - \Pr(\overline{A}_w)$,
%where $\overline{A}_w$ is the complement of $A_w$. To evaluate $\Pr(\overline{A}_w)$, we note that the number of subsequences of the sequence of permutation shifts corresponding to $w$ is upper bounded by $2^c$. By Lemma~\ref{lem25},
%each such subsequence corresponds to a TBC walk of zero permutation shift with probability at most $c/(4N)$. We thus have $\Pr(\overline{A}_w) < c2^c/(4N)$, which results in $\Pr(A_w ) > 1-c2^c/(4N)$. Thus, as $N$ tends to infinity, for a fixed $c$, we have
\begin{equation}
1 - \frac{c^3}{4N} \leq \Pr({\cal P}(w)=0) \times \Pr(A_w | {\cal P}(w)=0) \leq 1.
\label{ES1}
\end{equation}
For each TBC walk $w$ in ${\cal T}_2$, by the definition of ${\cal T}_2$,
there is a TBC subwalk $w'$ of $w$ such that ${\cal P}(w')= 0$.
So $\Pr(A_w)=0$, and thus
\begin{equation}
\Pr(\{{\cal P}(w)=0\} \cap A_w)=0.
\label{ES2}
\end{equation}
For each TBC walk $w$ in ${\cal T}_3$, using (\ref{eqsj}) and $0 \leq \Pr(A_w | {\cal P}(w)=0) \leq 1$, we have
%$\Pr({\cal P}(w)=0) \leq c/(4N)$, using Lemma~\ref{lem25}, and it is clear that $0 \leq   \Pr(A_w | {\cal P}(w)=0)\leq 1$. Therefore,
\begin{equation}\label{EZ1}
 0 \leq \Pr({\cal P}(w)=0) \times\Pr(A_w | {\cal P}(w)=0)\leq \frac{c}{4N} \:.
\end{equation}
%For each TBC walk $w$ in ${\cal T}_4(G,c)$, using (\ref{eqak}) and Lemma~\ref{lemtg}, we have
% \begin{equation}\label{EZ92}
% \frac{1}{N} (1-\frac{2^c}{N}) \leq \Pr({\cal P}(w)=0) \times\Pr(A_w | {\cal P}(w)=0)\leq \frac{1}{N}\:.
%\end{equation}
Replacing (\ref{ES1}), (\ref{ES2}), and (\ref{EZ1}) in (\ref{EC1}) completes the proof.
}\end{proof}

We note that, for a given base graph, the values $T_1$ and $T_3$ are fixed with respect to the lifting degree $N$. We thus have the following corollary, which demonstrates that the growth of the expected number of $c$-cycles with $N$ can follow two very different trajectories depending on the value of $c$ and whether the lifted graph has any inevitable cycle of length $c$ or not.
%In the following remarks, we study  the expected values of the number of cycles of a given length when the lifting degree is a prime number, and when the lifting degree is a large number.

\begin{cor}
Let  $\tilde{G}$ be a random cyclic $N$-lift of a base bipartite graph $G$. If $\tilde{G}$ contains inevitable cycles of length $c$ (i.e., graph $G$ contains at least one prime ZP TBC walk of length $c$), then,
%, and let $T(G,c)$ denote the number of TBC walks of length $c$ in $G$ .
as $N$ tends to infinity, the expected number of cycles of length $c$ in $\tilde{G}$ will be dominated by that of inevitable cycles and grows as $\Theta(N)$.\footnote{We use the notation $f(x)=\Theta(g(x))$, if for sufficiently large values of $x$, we have $a \times g(x) \leq f(x) \leq b \times g(x)$, for some positive $a$ and $b$ values.} On the other hand, if  $\tilde{G}$ contains no inevitable cycles of length $c$ (i.e., graph $G$ contains no prime ZP TBC walk of length $c$), then,
as $N$ tends to infinity, the expected number of cycles of length $c$ in $\tilde{G}$ is $\Theta(1)$ (is asymptotically constant with respect to $N$).
\label{cordw}
\end{cor}

\begin{rem}
It was shown in~\cite{kim2007cycle} that cyclic lifts $\tilde{G}$ of a base graph with girth $g$ and no parallel edges have no inevitable cycles of length smaller than $3g$.
Thus, based on Corollary~\ref{cordw}, for $c < 3g$, the expected number of cycles of length $c$ in $\tilde{G}$ is $\Theta(1)$.
\label{rembr}
\end{rem}

\begin{rem}
It is important to note the difference between the expected number of $c$-cycles of random lifts, discussed in Section~\ref{secN1}, and that of cyclic lifts, discussed in this section. While for random lifts, the expected value is $\Theta(1)$ with respect to lifting degree $N$, regardless of the value of $c$ or the base graph, for cyclic lifts, it can be $\Theta(N)$, depending on the value of $c$ and the base graph, as explained in Corollary~\ref{cordw}.
\label{remkw}
\end{rem}
%\begin{rem}{
%Let  $\tilde{G}$ be a  cyclic $N$-lifting of a base bipartite graph $G$ such that each permutation shift matrix is an identity matrix.  For any fixed even value $c \geq 4$, we have $E[N_c(\tilde{G})] =N\times C(G,c) $, where $C(G,c) $ is the number of cycles of length $c$ in the base graph (infact each cycle of length $c$ in the base graph will cause $N$ cycles of length $c$ in $\tilde{G}$). Thus, we have $E[N_c(\tilde{G})] = \mathcal{O}(N)$.  We know that by using permutation shift matrices we have $E[N_c(\tilde{G})] = \mathcal{O}(N)$. Therefore, using permutation shift matrices although make the derived graph strongly connected but it does not necessary decrease the order of $E[N_c(\tilde{G})]$ in terms of $N$. On the other hand, if instead of permutation shift matrices we use form permutation   matrices in the construction of the code, then by Theorem \ref{TofFortin} we have $E[N_c(\tilde{G})] =T(G,c)$, where $T(G,c)$ is the number of  TBC walks  of length $c$ in the base graph. Thus, the average number of cycles of length $c$ is $\mathcal{O}(1)$. Hence,  using permutation   matrices instead of permutation shift matrices decrease the order of $E[N_c(\tilde{G})]$ in terms of $N$.
%}\end{rem}

\subsection{Calculation of $\mathrm{Var}(N_c)$}

In the following, we prove that the variance of the number of cycles of length $c$ in a random cyclic $N$-lift increases at most linearly with $N$.

\begin{theo}
Let  $\tilde{G}$ be a random cyclic $N$-lift of a base bipartite graph $G$ with no parallel edges. As $N$ tends to infinity, for any fixed even value $c \geq 4$, we have
\begin{equation}
\mathrm{Var}[N_c(\tilde{G})]  \leq (\frac{c^3}{2} T_1^2+ \frac{c}{4} T_3^2+ \frac{c}{2} T_1 T_3) \times N + \mathcal{O}(1)\:.
\label{eqkd}
\end{equation}
%where $T_i$ is used as an abbreviation for $T_i(G,c)$.
\label{T5}
\end{theo}

\begin{proof}{
Following the same notations as in Theorem~\ref{Th2}, the number of cycles of length $c$ in $\tilde{G}$ is given by the following random variable:
$$
N_c(\tilde{G}) = N \sum_{w \in {\cal T}} I(\{{\cal P}(w)=0\} \cap A_w).
$$
%In Theorem \ref{Th2}, we showed that as $N$ tends to infinity, for any fixed even value $c \geq 4$, we have
%\begin{equation}\label{ES11}
%%E^2[N_c(\tilde{G})]= N^2\times T_1^2(G,c) + \Theta(1).
%\end{equation}
We have $\mathrm{Var}[N_c(\tilde{G})] = E[N^2_c(\tilde{G})] - E^2[N_c(\tilde{G})]$. In the following, we derive an upper bound on $E[N^2_c(\tilde{G})]$.
This together with the lower bound on $E^2[N_c(\tilde{G})]$, derived in Theorem~\ref{Th2}, will prove the theorem.
We have
% Using an argument similar to that made in the proof of Theorem~\ref{T2}, we have
\begin{align*}
E[N^2_c(\tilde{G})]  & = N^2 \sum_{w \in {\cal T}} \sum_{w' \in {\cal T}} E[I({\cal P}(w)=0 \cap A_w) I({\cal P}(w')=0\cap A_{w'})]\\
  & = N^2 \sum_{w \in {\cal T}} \sum_{w' \in {\cal T}} \Pr({\cal P}(w)=0 \cap A_w \cap {\cal P}(w')=0\cap A_{w'})\:. \numberthis
\label{ES6}
\end{align*}
%In the following, for simplicity of notations, we use ${\cal T}$ for ${\cal T}(G,c)$, and ${\cal T}_i$ for ${\cal T}_i(G,c)$.
To obtain an upper bound on $E[N^2_c(\tilde{G})]$, we break each of the two  summations in (\ref{ES6})
into three, each on one of the three partitions ${\cal T}_1, {\cal T}_2$, and ${\cal T}_3$ of ${\cal T}$.

Consider the case where $w \in {\cal T}_1$ and $w' \in {\cal T}_1$. In this case, we simply use the upper bound of one on $\Pr({\cal P}(w)=0 \cap A_w \cap {\cal P}(w')=0\cap A_{w'})$. This
contributes $T_1^2 \times N^2$ to the upper bound on the variance.

Now consider the case where $w \in {\cal T}_1$ and $w' \in {\cal T}_3$. In this case, we have
\begin{equation}
\Pr({\cal P}(w)=0 \cap A_w \cap {\cal P}(w')=0\cap A_{w'}) \leq \Pr({\cal P}(w')=0) \leq \frac{c}{4N}\:,
\label{eqf5}
\end{equation}
where the last inequality is from (\ref{eqsj}). Based on (\ref{eqf5}), the contribution of this scenario plus the case where $w \in {\cal T}_3$ and $w' \in {\cal T}_1$ in the upper bound is
$c /2 \times T_1 \times T_3 \times N$. Similarly, based on (\ref{eqf5}), the contribution of cases where $w \in {\cal T}_3$ and $w' \in {\cal T}_3$ is upper bounded by $c/4 \times T_3^2 \times N$.

%Following similar steps, and using (\ref{eqak}), the contribution of cases $w \in {\cal T}_1, w' \in {\cal T}_4$, and $w \in {\cal T}_4, w' \in {\cal T}_1$, in the upper bound, is $2 \times T_1 \times T_4 \times N$, the contribution of cases $w \in {\cal T}_3, w' \in {\cal T}_4$, and $w \in {\cal T}_4, w' \in {\cal T}_3$, is $2 \times T_3 \times T_4 \times N$, and the contribution of cases $w \in {\cal T}_4$ and  $w' \in {\cal T}_4$, is $T_4^2 \times N$.
For all the cases where either $w$ or $w'$ is in ${\cal T}_2$, we have $\Pr({\cal P}(w)=0 \cap A_w \cap {\cal P}(w')=0\cap A_{w'}) = 0$, and thus no contribution to the upper bound.

Adding up all the contributions of different cases, as discussed above, we obtain the following upper bound on $E[N^2_c(\tilde{G})]$:
$$
E[N^2_c(\tilde{G})] \leq T_1^2 \times N^2 + (\frac{c}{4} T_3^2 + \frac{c}{2} T_1 T_3) \times N\:.
$$
This combined with the lower bound of Theorem~\ref{Th2} on $E^2[N_c(\tilde{G})]$ complete the proof.
}\end{proof}

\section{Numerical results}
\label{sec4}
%In this section we focus on numerical results and we demonstrate how accurate our derivations are for finite lengths and that they in fact match the numerical results for larger block lengths.

\subsection{Random regular and irregular bipartite graphs}

In~\cite{karimi2013message}, the authors generated random codes from different bi-regular ensembles of LDPC codes, and empirically studied the distribution of cycles of
different length in such codes as a function of code's degree distribution and block length. The conclusion of~\cite{karimi2013message} was that the cycle distribution highly
depends on the degree distribution but does not change much with the block length $n$. In Corollary~\ref{C1}, we reached a similar conclusion through our theoretical analysis. In fact, we proved
that, in the asymptotic regime of $n \rightarrow \infty$, the cycle distributions are independent of $n$, and that the expected values of the number of $c$-cycles
increase polynomially with the node degrees and exponentially with the cycle length $c$.

In the following, we demonstrate through some examples that the expected values that we derived in Theorem~\ref{T01} and Corollary~\ref{C1}, match the numerical results. We start by the same examples
considered in Table IV of~\cite{karimi2013message}. The multiplicities of cycles of different lengths for rate-$1/2$ bi-regular codes of different degree distributions and lengths are
reproduced in Table~\ref{T1} here, and compared
%  randomly generated six parity check
%matrices for each of the following three rate-1/2
%LDPC code ensembles: $(d_u, d_w) = (3, 6), (4, 8), (5, 10)$. The lengths for each degree distribution are: $n =$
%200, 500, 1000, 5000, 10000 and 20000. In the generation of
%the parity-check matrices, 4-cycles were avoided. In Table 1 we compare the results with the expected values that
with the result of Corollary \ref{C1}. As can be seen, the expected values of Corollary~\ref{C1} are very close to the cycle multiplicities
of random realizations of the graphs for different block lengths, ranging from $200$ all the way to $20000$.

\begin{table}[ht]
\caption{Multiplicities of short cycles in the Tanner graphs of rate-$1/2$ random bi-regular LDPC codes with different degree distributions and different block lengths}
\begin{center}
\scalebox{1}{
\begin{tabular}{ |c|l||l|l|l|l|l|l||c|  }
\hline
Degree& Short Cycle &  \multicolumn{ 6}{  c| }{  Block Length} & $E[N_c]$  \\
\cline{3-8}
 Distribution &  Distribution & 200 & 500 & 1000 &  5000 &  10000  & 20000 & Corollary~\ref{C1}\\ \hline \hline
\multirow{3}{*}{(3,6)}
 & $N_6$    & 171  & 167   & 181  &  156 &  166  & 148  & 167\\
 & $N_8$    & 1265 & 1239  & 1226 &  1235&  1253 & 1285 & 1250\\
 & $N_{10}$ & 10069& 10110 & 9939 &  9982&  9858 & 9974 & 10000 \\ \hline

\multirow{3}{*}{(4, 8)}
 & $N_6$    & 1636  & 1611  & 1584  &  1562&  1537 & 1572 & 1544\\
 & $N_8$    & 25005 & 24419 & 24379 & 24363&  24529& 24557& 24310\\
 & $N_{10}$ & 409335& 409373& 408595&407958& 408246&409051& 408410 \\ \hline

\multirow{3}{*}{(5, 10)}
 & $N_6$    & 8626   & 8064  & 8055  &  7978 &  7858 & 7926 & 7776\\
 & $N_8$    & 213639 &212484 & 210767&210153 & 209614&210159& 209952\\
 & $N_{10}$ &6052158 &6054661&6049148&6043400&6049583&6043704& 6046617 \\ \hline

\end{tabular}
}
\end{center}
\label{T1}
\end{table}

%As we mentioned in the proof of Theorem \ref{T01}, in the  approximation  \ref{E102}, equality holds if and only if $d_1=d_2=\ldots=d_n$ and $d'_1=d'_2=\ldots=d'_{n'}$. Thus,  in the case of bi-regular graphs (and also regular graphs) the  approximation \ref{E92} is tight. Consequently, the final results are sharp. Someone can see how the expected values match with the real results in the above table.

As the next example, we consider two irregular degree distributions, and construct random codes of different block lengths with those degree distributions.
The first degree distribution is selected as $\lambda_I(x) = 0.4286x^2+0.5714x^3$, and $\rho_I(x) = x^6$, where the coefficients $\lambda_i$ and $\rho_i$ represent the fraction of edges connected to variable and check nodes of degree $i+1$, respectively.
This degree distribution, which is mildly irregular, corresponds to an LDPC code with rate $0.5$. We thus have $n = 2m$.
% for the corresponding Tanner graph, where $n$ is the block length, and $m$ is the number of check nodes in the Tanner graph.
%randomly selected the following irregular graph that was generated randomly. Let $G=(V\cup U,E)$ be a random irregular bipartite graph such that $|V|=n=1000$ and $|U|=n'=500$. The degree set of $V$ is $\{3,4 \}$ with multiplicity $2500,2500$, respectively. Also, the degree set of $U$ is $\{12,13,14,15,16,17 \}$ with multiplicity $3,41,529,520,94,13$, respectively. In this graph
The second degree distribution is selected from Table I of~\cite{saeedi2010design}. It is more irregular than the first degree distribution and is as follows:
$\lambda_{II}(x) = 0.2690x+0.2603x^2+0.0451x^4+0.4256x^9$, and $\rho_{II}(x) = 0.6398x^6+0.3602x^7$. The code rate corresponding to this degree distribution is $0.4998$~\cite{saeedi2010design}, and thus $n \simeq 2m$. In Table~\ref{T2}, we have provided the cycle multiplicities of the random realizations of the two degree distributions at block lengths $200$, $500$, $1000$, $5000$, $10000$ and $20000$, along with the approximation of expected values obtained based on the asymptotic upper bound of Theorem~\ref{T01}. Comparison of the results of Table~\ref{T2} with those of Table~\ref{T1} shows a larger discrepancy between the approximations of expected values and the cycle multiplicities in random realizations for irregular graphs vs. regular ones. This can be, at least in part, explained by Remark~\ref{Rx}. Moreover, comparison of the results for the two irregular degree distributions, particularly for the largest block length of $20000$, shows that the approximations provided for $E(N_c)$ by the asymptotic upper bound of Theorem~\ref{T01} are more accurate for the less irregular ensemble. 
%This is in agreement with the fact that the accuracy of approximation is directly related to $S(h_u) \times S(h_w)$, where larger irregularity in each side of the graph would imply larger values of $h_u$ and $h_w$, respectively, and consequently, larger values of $S(h_u)$ and $S(h_w)$. 
We also note that the asymptotic lower bounds for $E(N_c), c=4, 6, 8, 10$, corresponding to the irregular ensembles $I$ and $II$ are $52, 512, 5577, 64840$, and $10, 44, 210, 1077$, respectively. One can clearly see that the asymptotic upper bound provides a much more accurate estimate for the number of cycles of different length in comparison with the asymptotic lower bound derived in Theorem~\ref{T01}. This is particularly the case for the more irregular degree distribution.

\begin{table}[ht]
\caption{Multiplicities of short cycles in the Tanner graphs of irregular LDPC codes with different degree distributions and different block lengths}
\begin{center}
\scalebox{1}{
\begin{tabular}{ |c|l||l|l|l|l|l|l||c|  }
\hline
Degree& Short Cycle &  \multicolumn{ 6}{  c| }{  Block Length} & $E[N_c]$  \\
\cline{3-8}
 Distribution &  Distribution & 200 & 500 & 1000 &  5000 &  10000  & 20000 & Theorem~\ref{T01}\\ \hline \hline
\multirow{4}{*}{$\lambda_I(x), \rho_I(x)$}
 & $N_4$    & 56   & 62    & 61   &  52  &  61   & 59   & 59\\
 & $N_6$    & 599  & 602   & 587  &  590 &  597  & 602  & 611\\
 & $N_8$    & 6653 & 6814  & 6742 &  6881&  7011 & 7158 & 7067\\
 & $N_{10}$ &85244 & 87260 & 84846& 86436& 87046 & 87311& 87181 \\ \hline

\multirow{3}{*}{$\lambda_{II}(x), \rho_{II}(x)$}
 & $N_4$    & 230   & 222    & 244   & 236  &  243  & 196  & 225\\
 & $N_6$    & 4871  & 4759  & 4057  &  4571&  4562 & 4769 & 4500\\
 & $N_8$    & 109017 &107523   &104599 &106620& 105685&107479& 101250\\
 & $N_{10}$ &2610260 &2557357 &2212847&2585699&2548117&2605595& 2430000 \\ \hline

\end{tabular}
}
\end{center}
\label{T2}
\end{table}

\subsection{Random lifts of a base graph}

We consider random lifts of the $3 \times 5$ fully-connected base graph with lifting degrees $400$, $1000$ and $2000$. The cycle multiplicities of the random lifts for cycles of length $4$ all the way to $16$ are shown in Table~\ref{TTT1}, and compared with the expected value obtained from Theorem~\ref{T3}. As can be seen, for different lifting degrees, the expected value provides a good approximation for the multiplicities of cycles of different length
in random realizations.

\begin{table}[ht]
\caption{Multiplicities of short cycles of different length for random lifts of different degrees of the
 $3 \times 5$  fully-connected base graph}
\begin{center}
\scalebox{1}{
\begin{tabular}{ |c |l|l|l||c |  }
\hline
Cycle  &  \multicolumn{ 3}{  c| }{  Lifting Degree} &   $E[N_c]$ \\
\cline{2-4}
 Length & $N=400$ & $N=1000$ & $N=2000$   &  Theorem~\ref{T3}\\ \hline \hline

  4    & 31    & 27     & 29      &   30 \\
  6    & 64    & 62     & 66      &   60 \\
  8    & 590   & 588    & 515     &   585\\
  10   & 2994  &3111    &3083     &  3060\\
  12   &22730  &22636   & 22919   &   22550\\
  14   &147395 &148141  &147894   &   147420\\
  16   &1058149&1061667 &1052401   &  1056832\\ \hline

\end{tabular}
}
\end{center}
\label{TTT1}
\end{table}

\subsection{Random QC  bipartite graphs}

In~\cite{karimi2012counting}, the authors studied the cycle distribution of random cyclic lifts of the $3 \times 5$ fully-connected base graph for different lifting degrees (block lengths), and
observed that such graphs have generally larger girth compared to random bi-regular codes with the same degree distribution and block length. The example also showed that the girth of QC codes was improved by the increase in the lifting degree $N$. The above results reported in Table I of~\cite{karimi2012counting} are reproduced here in Table~\ref{TTT3}.

We note that the $3 \times 5$ fully-connected  base graph has girth $4$, and thus, based on Remark~\ref{rembr}, for $c \leq 10$, cyclic random lifts of this base graph have no inevitable cycles
of length $c$. This means that for $c \leq 10$, the expected value of the number of cycles of length $c$ does not increase with the lifting degree $N$. On the other hand, one
can find prime ZP TBC walks of length $12$, $14$ and $16$ in the base graph: let $G=(U,W)$ be the $3 \times 5$ fully-connected base graph with $U=\{1,2,3\}$ and $W=\{4,5,6,7,8\}$.
It is then easy to verify that the following TBC walks in $G$ are prime with zero permutation shifts: ${w}_{12}=5243514253415$, ${w}_{14}= 342536143524163$ and ${w}_{16}=25362714263524172$.
This means that the random cyclic lifts of the base graph will have inevitable cycles with these lengths and that, based on Corollary~\ref{cordw}, the expected value of cycles with these lengths increases linearly with $N$ for sufficiently large $N$ values. These theoretical predictions are consistent with the numerical results reported in Table~\ref{TTT3}, for these cycle lengths. For cycles of length $18$, however, there is no prime ZP TBC walk in the $3 \times 5$ fully-connected base graph, and thus the expected number of such cycles remains constant with respect to $N$. This is also consistent with the results of Table~\ref{TTT3}.

For comparison, we have also included,  in the last column of Table~\ref{TTT3}, the expected value of the number of cycles in random lifts of the $3 \times 5$ fully-connected base graph, obtained based on Theorem~\ref{T3}.
One can see the large difference between these values and the corresponding values for random cyclic lifts for cases of $c=12, 14,$ and $16$, where the cyclic lifts have inevitable cycles.

%In Section~\ref{sec3}, we demonstrated through Theorem~\ref{Th2} and Remark \ref{R55}, that when the base graph is a  $3\times 5$ fully-connected graph, then like what the experiment of~\cite{karimi2012counting} may suggest,  the expected value of the number of cycles of  length $c$, $12 \leq c \leq 16$, in a cyclic lifting of a base graph  depends on $N$. But, for $c=6,8,10$  the expected values of the number of cycles of those lengths do not increase by $N$.  We also proved in Theorem~\ref{T5} that, like the expected value, the variance of the number of cycles is a linearly increasing function of $N$, thus providing another justification for the numerical results of \cite{karimi2012counting} and especially for the number of cycles of length 10 in Table \ref{T3}.

\begin{table}[ht]
\caption{Multiplicities of cycles of different length for random cyclic lifts of different degrees of the
 $3 \times 5$  fully-connected base graph}
\begin{center}
\scalebox{1}{
\begin{tabular}{ |c |l|l|l||c |  }
\hline
Cycle  &  \multicolumn{ 3}{  c| }{  Lifting Degree} &   $E[N_c]$ \\
\cline{2-4}
 Length & $N=400$ & $N=1000$ & $N=2000$   &   Theorem~\ref{T3}\\ \hline \hline

  6    & 0    & 0     & 0      &   60 \\
  8    & 0    & 0     & 0      &   585 \\
  10   & 2000 & 1000  & 0      &   3060\\
  12   & 33200&54000  &98000   &   22550\\
  14   &193200&275000 & 478000 &   147420\\
  16   &1022200&1169000&1490000&   1056832\\
  18   &7143600&7251000&8282000&   7427300\\ \hline

\end{tabular}
}
\end{center}
\label{TTT3}
\end{table}

\section{CONCLUSION}
\label{sec5}

In this paper, we studied the cycle distribution of different ensembles of LDPC codes, often used in the literature, in the asymptotic regime where the block length tends to infinity (but the degree distribution is fixed). These ensembles were random irregular and bi-regular, random lifts of protographs, and random cyclic lifts of protographs. We demonstrated that for the first ensemble, the multiplicities of cycles of different lengths have independent Poisson distributions. We derived asymptotic upper and lower bounds on the expected values of the distributions. These bounds are only a function of cycle length and degree distributions, and independent of the block length. We also showed that for the second ensemble, the
asymptotic cycle distributions have the same behavior as those of the first ensemble as long as the degree distributions are identical. For the third ensemble, we proved that the
 cycle distributions can be significantly different than those of the first two ensembles. In particular, we showed that for some values of $c$, and depending on the protograph,
 the expected number of $c$-cycles can increase linearly with the block length. We also derived an upper bound, linearly increasing with the block length, on the variance of the number of $c$-cycles.

Using numerical results, we demonstrated that our asymptotic results provide good approximations for the number of cycles in realizations of finite-length LDPC codes, even when
the block length is as short as a few hundred bits. Moreover, our results provided theoretical justification for some of the observations made empirically in the literature about cycle distributions of LDPC codes.

The results presented in this paper can be used in the analysis and design of LDPC codes in cases where such processes depend on the knowledge of the cycle distributions. As a particular example, we showed how the asymptotic average number of trapping sets can be estimated using the results presented in this work.

Finally, our numerical results show that for irregular graphs, the asymptotic upper bound provided in Theorem~\ref{T01} on the expected value of cycle multiplicities is much more accurate than the asymptotic lower bound in estimating the cycle multiplicities. This suggests that it may be possible to tighten the asymptotic lower bound derived in Theorem~\ref{T01}. 

\section{Acknowledgment}
The authors wish to thank the anonymous reviewers whose comments improved the presentation of the paper.

\section{Appendix I}
\label{App1}

{\bf Proof of Lemma~\ref{lemxz}.} Let $|V(H)|$ and $|E(H)|$ be the number of nodes and the number of edges of $H$, respectively, and let 
${\cal C}_H$ be the number of structures in a random configuration whose pojections in ${\cal G}^*$ are copies of $H$. There are at most $\mathcal{O}({(\Delta+1) n \choose |V(H)|})$ choices for the node set of the copy of $H$. 
Thus, we have ${\cal C}_H=\mathcal{O}(n^{|V(H)|})$. On the other hand, the probability of each given edge set of size $|E(H)|$ is $\dfrac{ (\eta-|E(H)|)!}{\eta!}=\mathcal{O}(n^{-|E(H)|})$. 
Thus, the expected number of copies of $H$ in ${\cal G}^*$ is $\dfrac{{\cal C}_H \times(\eta-|E(H)|)!}{\eta!}=\mathcal{O}(n^{|V(H)|-|E(H)|})=\mathcal{O}(\frac{1}{n})$. 
%This completes the proof of claim. 
$\blacksquare$

{\bf Proof of Lemma~\ref{lemcx}.} Consider the left hand side of (\ref{NewEE4}). There are ${ n \choose k}$ terms added together, each being a product of $k$ distinct variables from the set $X$, and each with the multiplicative coefficient ${k \choose k/2}$.
%(x_{i_1} x_{i_2} \ldots x_{i_k}) $,  and no term has a repeated variable. 
This implies that on the left side, we have ${ n \choose k} \times {k \choose k/2} = \Theta(n^{k})$ terms, each a product of $k$ distinct variables from the set $X$, added together. 
%So in the left side of (\ref{NewEE4}), the degree of each term is exactly $k$, the coefficient of each term is ${k \choose k/2}$ and the polynomial  has ${ n \choose k}$   terms. 
Now, consider the right hand side of (\ref{NewEE4}). It is the product of two identical expressions, each a sum of ${ n \choose k/2}$ terms, where each such term is a product of $k/2$ distinct variables from the set $X$.
If we expand the product of the two identical expressions, we have the sum of ${ n \choose k/2} \times { n \choose k/2}$ product terms, where each product involves $k$ variables from the set $X$. We can partition such product terms into two categories: ($1$) those with all $k$ variables being distinct, and ($2$) those with at least one variable repeated at least once. In the following, we show that the first category consists of exactly the same product terms as in the left hand side of (\ref{NewEE4}), and that the second category contains $\mathcal{O}(n^{k-1})$ terms. This will then prove the asymptotic equality of (\ref{NewEE4}). 

On the right hand side of (\ref{NewEE4}), the number of product terms in Category $1$ is equal to  ${ n \choose k/2} \times { n-k/2 \choose k/2}$. The term ${ n \choose k/2}$ is the number of product terms of size $k/2$ in the first expression, and the  term ${ n-k/2 \choose k/2}$ is the number of product terms in the second expression that have no common variable with the selected product term from the first expression. It is now easy to see that considering all the possible ${ n \choose k}$ product terms with $k$ distinct variables, each is repeated ${k \choose k/2}$ times in the product terms of Category $1$. In fact, we have ${ n \choose k/2} \times { n-k/2 \choose k/2} = { n \choose k} \times {k \choose k/2}$. Now, the number of terms in Category 2 is equal to ${ n \choose k/2} \times { n \choose k/2} - { n \choose k/2} \times { n-k/2 \choose k/2}$, which is $\mathcal{O}(n^{k-1})$.
$\blacksquare$

\bibliographystyle{ieeetr}
\bibliography{POref}

\end{document}